\documentclass[12pt]{article}
\usepackage{fullpage}
\usepackage{amsmath}
\usepackage{amsthm}
\usepackage{amssymb}
\usepackage{abstract}
\usepackage{tensor}
\usepackage[hyperfootnotes=true,colorlinks=true,linkcolor=blue,citecolor=blue,urlcolor  = blue,anchorcolor = blue,linktocpage]{hyperref}
\numberwithin{equation}{section}
\usepackage[width=\textwidth]{caption}
\usepackage{cite} 
\usepackage[utf8]{inputenc}
\usepackage{tikz}
\usepackage[percent]{overpic}
\usepackage{float}
\usepackage{adjustbox}
\restylefloat{table}

\newcommand{\be}{\begin{equation}}
\newcommand{\ee}{\end{equation}}
\newcommand{\ext}[1]{\mathbf{d}\ensuremath{\mathbf{#1}}}

\newtheorem{theorem}{Theorem}[section] 

\newtheorem{lemma}{Lemma}[section]

\title{Bianchi models with a free massless scalar field: \\ 
	invariant sets and higher symmetries}

\date{\vspace{-5ex}}

\author{  
	Mikjel Thorsrud
	\\
	\small{\em Faculty of Engineering, \O stfold University College,}\\
	\small{\em P.O. Box 700, 1757 Halden, Norway.} \\
	\small{E-mail: mikjel.thorsrud@hiof.no }
}

\begin{document}

\maketitle

\begin{abstract}		
	We scrutinize the overall structure of the space of cosmological models of Bianchi type I-VII$_h$ that contain a free massless scalar field, with a spatially homogeneous gradient $\nabla_\mu \varphi$ that generally breaks isotropy, in addition to a standard perfect fluid. Specifically, state space is written as a union of disjoint invariant sets, each corresponding to a particular cosmological model that is classified with respect to the Bianchi type and the matter content. Subsets corresponding to higher-symmetry models, including all locally rotationally symmetric models and FLRW models, and models with a shear-free normal congruence, are also derived and classified. For each model the dimension ($d$) of the space of initial data is given, after fixing the orientation of the orthonormal frame uniquely relative to matter anisotropies and geometrical anisotropies. 
\end{abstract}


\newpage
\setcounter{tocdepth}{3}
\tableofcontents
\newpage

\section{Introduction}
In mathematical cosmology there is a vast literature on spatially homogeneous models beyond the Friedmann-Lema\^{i}tre-Robertson-Walker (FLRW) family of metrics. Efforts initiated by G.F.R. Ellis and M.A.H. MacCallum in the late 60's \cite{Ellis:1967, main:Ellis1969, MacCallum:1973} culminated in the orthonormal frame formalism, which has long been a standard approach to dynamical systems in cosmology \cite{bok:EllisWainwright}. 
The qualitative behaviour of Bianchi cosmological models containing perfect fluids is generally well understood \cite{Hsu-Wainwright:1986, Wainwright-Hsu:1989, Hewitt-Wainwright:1993}, even when the fluid four-velocity is tilted relative to the normal congruence $u^\mu$.\footnote{Spatially homogeneous cosmological models posses a unique future oriented timelike unit vector $u^\mu$, the \emph{normal congruence}, which is orthogonal to hypersurfaces of homogeneity. For a complete list of references to studies of \emph{tilted} perfect fluid models see table 18.4 of \cite{bok:EllisMaartensMacCallum}.} In non-tilted models the perfect fluid couples off from the shear evolution equations and thus only sources the expansion scalar.  The dynamical behavior of such models, which is still rich and highly non-trivial, can therefore be seen as a result of geometrical self-sourcing. In order to see general relativity's intimate relationship between geometry and matter in action in a cosmological model, however, the assumption of three-dimensional rotational invariance must be relaxed in the matter sector as well as in the spacetime geometry. In contrast to perfect fluid models, little is known about Bianchi models containing isotropy violating matter sources \cite{Calogero:2009, Hervik:2011}. To find natural candidates for such matter fields is not difficult; in fact a free massless scalar field will be equipped with an energy-momentum tensor that is generally imperfect when decomposed relative to the normal congruence $u^\mu$.

Scalar fields have a major role in physics, from high-energy physics models such as supergravity and string theories to the standard model of particle physics. In theoretical cosmology scalar fields play a major role in model building beyond the $\Lambda$CDM concordance model \cite{Copeland:2006}. The simplest cases are so-called quintessence models with a Lagrangian of the form 
\be
\mathcal L = -\frac{1}{2} g^{\mu\nu}\nabla_\mu \varphi \nabla_\nu \varphi - V(\varphi) \,.
\ee
A free and massless scalar field, $V(\varphi)\rightarrow 0$, belongs to a family of free gauge fields with action on the form
\begin{equation}
\label{action2}
S_{\mathcal A}=-\frac{1}{2}\int{\boldsymbol{\mathcal{F}}}\wedge \boldsymbol{\star\mathbf{\mathcal{F}}},\qquad \boldsymbol{\mathcal{F}} \equiv  \ext{\boldsymbol{\mathcal A}} \,,
\end{equation}
where $\boldsymbol{\mathcal A}$ is a $p$-form gauge field \cite{bok:supergravity}. In spacetime there are only three independent possibilities. The case $p=3$ corresponds to a cosmological constant, whereas $p=1$ corresponds to a Maxwell type field. The latter has received some attention in the past in the context of source-free cosmological magnetic fields \cite{LeBlanc:1997,Yamamoto:2011}. The remaining cases $p=0$ and $p=2$ are equivalent upon Hodge dual (at the field strength $p+1$ level), and correspond to a free massless scalar field with Lagrangian $\mathcal L = -\tfrac{1}{2}g^{\mu\nu}\nabla_\mu \varphi \nabla_\nu \varphi$  upon identifying $\nabla_\mu \varphi$ as the components of the field strength $\boldsymbol{\mathcal{F}}$. 

Among gauge fields in the family (\ref{action2}) only the case of a cosmological constant breaks the strong energy condition. A free and massless scalar field is of course not capable of driving an accelerated phase of expansion. If such a scalar field is homogeneous on spatial sections, $\varphi=\varphi(t)$, it merely behaves as a non-tilted ``stiff'' perfect fluid, which decays away rapidly with the expansion of the universe. It is thus easy to overlook the interesting role such a simple and natural matter field could have in cosmological model building in general and in the physics of the early universe in particular, where the phenomenological role of various isotropy violating fields have been devoted a great deal of attention (see \cite{Ford:1989, Ackerman:2007, Golovnev:2008, Watanabe:2009,Hervik:2011, Maleknejad:2012,Ito:2015,Almeida:2019,Cicciarella:2019} and references therein) often in order to address the so-called $\Lambda$CDM anomalies \cite{Bennett:2010}.  Here it is important to note that within spatially homogeneous cosmological models, in the absence of a potential function $V(\varphi)$, the only symmetry condition on the scalar field is that its gradient is spatially homogeneous. Unless $\nabla_\mu \varphi$ is parallel to the normal congruence $u^\mu$, such a scalar field breaks the isotropy of spatial sections. It will thus generally serve as an imperfect matter model, with a dynamical effective equation of state parameter ranging from $-1/3$ to $1$, that must be embedded in an anisotropic universe model.  

When $\nabla_\mu \varphi$ is orthogonal to $u^\mu$ the effective equation of state parameter is fixed to $-1/3$ and the scalar field evolve, informally, as an ``anisotropic spatial curvature field''. Such a field may be arranged to counter-balance the geometrical spatial curvature so that the anisotropies freeze out and become non-dynamical. Carneiro and Marug\'an presented such an exact shear-free solution in the locally rotationally symmetric (LRS) Bianchi type III metric \cite{main:Carneiro01}. This solution possess a shear-free normal congruence, so the expansion of the universe is isotropic, despite the intrinsic anisotropy of the underlying Bianchi model. Such solutions are dynamically equivalent to standard FLRW cosmological models and could have rather interesting phenomenological consequences \cite{main:Koivisto11}, that remain mainly unexplored. In \cite{Mik18} it was shown that among matter models of the family (\ref{action2}), only the case of a free and massless scalar field allows construction of cosmological models beyond FLRW with a shear-free normal congruence.  See start of \ref{sub:shear-free} for more details on studies of shear-free models. In \cite{Ben18} the orthonormal frame approach were adopted in order to study the general dynamics of cosmological models containing a non-tilted $\gamma$-law perfect fluid in addition to a $p$-form gauge field with $p\in\{0,2\}$.  A detailed dynamical system analysis was carried out for Bianchi type I and V, which revealed that the relevance of a free massless scalar field in a cosmological context is not limited to shear-free models. Especially in Bianchi type I a dynamically very rich behavior was identified, providing interesting features to be explored in the context of the early universe. Specifically, for each $\gamma$ in the range $(6/5, 2]$, which includes the radiation case $\gamma=4/3$, there was shown to be a unique non-LRS self-similar solution with a rotating field strength $\boldsymbol{\mathcal F}$ that is dynamically stable, in fact a global attractor of the type I state space. 

Overall, self-similar solutions found and explored in models with shear-free dynamics \cite{main:Carneiro01, main:Koivisto11, Mik18} and general dynamics \cite{Ben18}  motivate a more systematic study of this family of cosmological models. In principle we are interested in all Bianchi type I-IX models, which are the most general spatially homogeneous models that contain FLRW universes as special cases. However, in Bianchi type VIII and IX the equations for the scalar field allow only a temporal component of $\nabla_\mu \varphi$ \cite{Ben18}. Including these two cases would therefore not provide any information beyond what is already known from studies of Bianchi models with perfect fluids.

\newpage
In this paper the goal is to systematically investigate the overall structure of the space of cosmological models of Bianchi type I-VII$_h$ that contain two matter components with separately conserved energy-momentum tensors:  
\begin{enumerate}
	\item[1)] a free massless scalar field with a spatially homogeneous gradient $\nabla_\mu \varphi$ and
	\item [2)] a perfect fluid with a linear barotropic $\gamma$-law equation of state and with four-velocity orthogonal to hypersurfaces of homogeneity.  
\end{enumerate}
The main question we address is: \emph{What type of cosmological models are permitted by general relativity in this setup?} We address it by deriving all cosmological models subject to the aforementioned assumptions and classifying them with respect to the matter content, the Bianchi type and higher symmetries. Our starting point is the field equations expressed as a constrained autonomous system in expansion normalized variables. We shall use the standard 1+1+2 decomposition of Einstein's field equations admitted by Bianchi type I-VII$_h$ models \cite{bok:EllisWainwright,Coley-Hervik:2005}, first employed for this class of imperfect matter models in \cite{Ben18} but the orientation of the orthonormal frame will be fixed uniquely relative to matter anisotropies and geometrical anisotropies from the start, as established in section \ref{ch:model}. 

Our investigation of cosmological models is based on the identification of invariant sets of the constrained autonomous system. An invariant set is here understood as a subset of state space $D$ which is invariant under time evolution; i.e. if the state vector $X(\tau)$ belongs to a set $\mathcal B$ at a certain time $\tau_0$, then it belongs to $\mathcal B$ at all times $\tau$.\footnote{See \cite{bok:EllisWainwright} for an introduction to the concept of invariant sets in the context of dynamical systems.} Specifically, we write state space $D$ as a union of disjoint invariant sets, each of which defines a unique cosmological model of a given Bianchi type and with a given matter model. Each set belongs either to the perfect branch, where the matter sector is isotropic, or the imperfect branch, where $\nabla_\mu \varphi$ breaks the isotropy. Among the sets we derive are therefore all the familiar perfect-fluid cosmological models, which are well-known from the literature. This has the advantages that, first, it is easy to check that our results agrees with standard results in the isotropic limit and, second, it allows us to observe extensions of familiar cosmological models into the imperfect branch. 

For each invariant set we report the following information:
\begin{itemize}
	\item The defining conditions for each cosmological model. These conditions are general and gives all representations in state space compatible with our choice of orthonormal frame established in section \ref{ch:model}. 
	\item The Bianchi type (I-VII$_h$), subsets with higher symmetries (LRS and FLRW models) and expansion symmetries (shear-free models). In LRS models the direction of the symmetry axis is given explicitly. 
	\item The matter model (perfect branch, imperfect branch, direction of energy flux).
	\item The dimension $d$ of the space of initial data. This is the number of independent constants that must be specified on the Cauchy surface, which is a measure of the generality of each model. 
\end{itemize}

For a summary of main findings, see table \ref{tab:sets:C} and \ref{tab:sets:D} for general sets and table \ref{tab:sets:LRS} for particular sets with higher symmetries. 

\paragraph{Organization of paper}
The paper is organized in the following way. In section \ref{ch:model} we fix the orientation of the orthonormal frame and present the underlying constrained autonomous system that is the starting point of the investigation. In section \ref{ch:invariant} state space is written as a union of disjoint invariant sets, each of which is identified as a particular cosmological model. In section \ref{ch:higher} all higher-symmetry subsets are derived that correspond to LRS and FLRW models. In section \ref{ch:ex-sym} subsets with expansion symmetries are identified, including all those with a shear-free normal congruence. Finally, section \ref{ch:summary} gives a summary with concluding remarks. The first part of this last section may be useful for readers not familiar with the concept of gauge fixing in the context of the orthonormal frame formalism. 

\paragraph{Conventions} Greek indices ($\alpha$, $\beta$, $\dots$) run from 0 to 3, Latin indices ($a$, $b$, $\dots$) from 1 to 3. Units satisfy $c=1$ and $8\pi G=1$.

\section{Model \label{ch:model}}
In this section we introduce the cosmological model and establish the constrained autonomous system which is presented in section \ref{sub:evo}. This dynamical system will be the starting point of our investigation of cosmological models. Our approach starts with the standard 1+1+2 decomposition of Einstein's field equations admitted by Bianchi type I-VII$_h$ models \cite{bok:EllisWainwright,Coley-Hervik:2005,Ben18}, but we fix the orientation of the orthonormal frame uniquely relative to matter anisotropies and geometrical anisotropies.

\subsection{Geometry \label{sub:geometry}}
In Bianchi cosmological models the metric admits a three-dimensional group $G_3$ of isometries that acts simply transitively on hypersurfaces of homogeneity \cite{bok:EllisWainwright}. Conveniently, the Bianchi models that are relevant for us to explore, are exactly those that admit a two-dimensional Abelian subgroup $G_2$ of isometries, which is Bianchi type I-VII$_h$ models.\footnote{As mentioned in the introduction, in Bianchi type VIII and IX the field equations for the scalar field restrict  $\nabla_\mu \varphi$ to be parallel to $u^\mu$, in which case the scalar field is dynamically trivial and equivalent to a stiff perfect fluid comoving with the normal congruence.} Adopting the orthonormal frame approach, the metric can be written 
\begin{equation}
\boldsymbol{{\rm ds}^2} =-\boldsymbol{\rm d}t^2+\delta_{ab}\,\boldsymbol{\omega}^a\boldsymbol{\omega}^b \,,
\label{line-element}
\end{equation}
where $\{\boldsymbol{\omega}^\mu\}=\{\boldsymbol{\rm d}t,  \boldsymbol{\omega}^a\}$ is an orthonormal frame (tetrad) obeying
\begin{equation} 
\ext{\boldsymbol{\omega}}^a=-\frac{1}{2}\tensor{\gamma}{^a_{bc}}(t)\boldsymbol{\omega}^{b}\wedge\boldsymbol{\omega}^c-\gamma^a_{~0c}(t)  \boldsymbol{{\rm d}}t\wedge\boldsymbol{\omega}^c \,,
\label{dw}
\end{equation}
which is dual to $\{\boldsymbol{e}_\mu\} = \{ \frac{\partial}{\partial t}, \mathbf e_a \}$. Here the timelike basis vector $\mathbf e_0$ is chosen orthogonal to spatial sections of homogeneity ($t=$ constant). 

We define the Hubble parameter $H$ and rate of shear tensor $\sigma_{\mu\nu}$ relative to the normal congruence $u_\mu=-\partial_\mu t$. Recall that the shear tensor is symmetric $\sigma_{\mu\nu}=\sigma_{(\mu\nu)}$, trace-free $\sigma^\mu_{\;\;\mu}=0$ and satisfies $u^\mu\sigma_{\mu\nu}=0$ so that only the spatial components $\sigma_{ab}$ are non-zero in the considered frame.  Since the vorticity and acceleration of the normal congruence is always zero \cite{main:Ellis1969}, $H$ and $\sigma_{\mu\nu}$ are given by
\be
\nabla_a u_b = H \delta_{ab} + \sigma_{ab} \,.
\ee
The structure  coefficients depend only on time and the spatial components are decomposed as 
\begin{eqnarray}
\label{strcoeff}
\tensor{\gamma}{^m_{ab}}&=\tensor{\epsilon}{_{abn}}n^{nm}+a_a\tensor{\delta}{_b^m}-a_b\tensor{\delta}{_a^m}\,,
\end{eqnarray}  
where $a_i$ is a covector, $n^{ab}$ is a symmetric tensor density and $\epsilon_{abc}$ is the Levi-Civita symbol ($\epsilon_{123}=1$). The evolution equations are \cite{bok:EllisWainwright}
\begin{align}
\dot a_i &= - H a_i - \sigma_i^{\;\;j} a_j + \epsilon_i^{\;\;jk}a_j\Omega_k \,, \label{eq:a} \\
\dot n_{ij} &= - H n_{ij} + 2\sigma_{(i}^{\;\;\; k} n_{j)k} + 2\epsilon^{mn}_{\;\;\;\;\; (i} n_{j)m} \Omega_n \,, \label{eq:n}
\end{align}
where the ``dot'' is a derivative with respect to proper time $t$ and $\Omega_i$ is the local angular velocity of the spatial frame $\{\mathbf e_a\}$ relative to a Fermi-propagated spatial frame. 

Following \cite{bok:EllisWainwright} we now perform a 1+2 decomposition of the spatial structure coefficients. First choose an orientation of the spatial frame $\{\mathbf e_a\}$ so that 
\be
a_i=(a,0,0) \,.
\label{def:a}
\ee
From (\ref{eq:a}) it follows that the frame must rotate in the following way:
\begin{align}
\Omega_2 &= \sigma_{13} \, , \label{Omega2}\\  
\Omega_3 &= -\sigma_{12} \, . 
\label{Omega3}
\end{align}
This still leaves an unused gauge degree of freedom associated with rotation around $\mathbf e_1$ that will be fixed relative to the imperfect matter sector in the next subsection.

Since $a_i$ is in the Kernel of $n^{ab}$ the gauge choice (\ref{def:a}) implies $n^{1a}=0$ in class B models ($a\neq 0$). This choice is also possible in those Bianchi class A models ($a=0$) that admits an Abelian subgroup $G_2$ of isometries. It is easy to check that the condition $n^{1a}=0$ is preserved in time from the evolution equation (\ref{eq:n}) given the gauge choice (\ref{Omega2})-(\ref{Omega3}). The choice $a_i=(a,0,0)$ and $n^{1i}=0$ thus exclude only Bianchi type VIII and IX models.  The result is a 1+1+2 decomposition with $\mathbf e_2$ and $\mathbf e_3$ tangent to the orbits of the $G_2$ subgroup.

\subsection{Matter fields \label{sub:matter}}
In the matter sector we consider a free massless scalar field with a spatially homogeneous, generally anisotropic, gradient $\nabla_\mu \varphi$ and a $\gamma$-law perfect fluid. Only the perfect fluid is assumed to be comoving with the normal congruence $u^\mu$.

\paragraph{Free massless scalar field}
In the language of differential forms the action of a free and massless scalar field is
\begin{equation}
\label{action}
S_{\varphi}=-\frac{1}{2} \boldsymbol{\int \ext\varphi \wedge \boldsymbol{\star}(\ext\varphi)} \,.
\end{equation}
The Maxwell type equations $\ext{\boldsymbol{\ext\varphi}}=0$ and $\ext{\boldsymbol{\star \ext\varphi}}=0$ follow, the first one as an identity. We next write the components of the gradient as
\be
\boldsymbol{ \ext\varphi} = x_\mu(t) \boldsymbol{\omega}^\mu \,, \quad x_\mu\equiv \nabla_\mu \varphi \,.
\ee
Our only assumption is that $x_\mu$ are functions of time $t$ only. We shall often refer to $x_\mu$ as the \emph{field strength one-form}, since it is the components of the field strength ($\boldsymbol{\mathcal F}$) in the gauge field description (\ref{action2}). The ``Maxwell equations'' can then be written on component form, relative to the orthonormal frame, as 
\begin{align} 
\partial_\mu x_\nu - \partial_\nu x_\mu  + \gamma^\gamma_{\;\;\nu\mu}x_\gamma = 0 \,, \label{eqm:dF} \\   \partial_\alpha x^\alpha  + x^\beta\gamma^\alpha_{\;\;\beta\alpha} = 0 \,, \label{eqm:dsF}
\end{align} 
where $\partial_\mu$ denotes the directional derivative along the basis vector $\mathbf e_\mu$. The first equation follows since $x_\mu$ is an exact 1-form by definition (so it must be closed as well), whereas the second equation is the usual Klein-Gordon equation for the scalar field. 

In order to fix the remaining gauge degree of freedom we write down the $(\mu, \nu)=(0, i)$ components of (\ref{eqm:dF}) explicitly:
\be
\dot x_i = - H x_i - \sigma_i^{\;\;j} x_j + \epsilon_i^{\;\;jk}a_j\Omega_k \,.
\label{eq:x}
\ee
Note that the form of this equation is identical to the form of the evolution equation (\ref{eq:a}) for $a_i$. It follows that the frame can be aligned with field strength one-form so that the conditions $x_2=x_3=0$ are preserved in time with the same frame rotation as given in (\ref{Omega2}) and (\ref{Omega3}). Generally $a_i$ and $x_i$ are not aligned though, but the gauge choice $a_i=(a,0,0)$ allows for choosing an orientation of the frame so that $x_i=(x_1,0,x_3)$ at any given instant of time. According to (\ref{eq:x}) the condition $x_2=0$ is preserved in time if 
\be
\Omega_1 = \sigma_2^{\;\;3} \,.
\ee  
The gauge fixing will be summarized and completed in the following subsection.

\paragraph{Perfect fluid}
Beside the scalar field we also consider a perfect fluid  obeying a barotropic linear $\gamma$-law equation of state
\begin{equation}
\label{eospf}
P=(\gamma-1)\rho\,, 
\end{equation} 
where $\gamma\,\in[0,2]$ is a constant. The four-velocity of the fluid is identified with the normal congruence $u^\mu$ by the following choice of energy-momentum tensor:
\be
T^{\mu\nu} = \rho u^\mu u^\nu + P (g^{\mu\nu}+u^\mu u^\nu) \,.
\ee
Hence, the Hubble parameter $H$ and shear tensor $\sigma_{\mu\nu}$ measure the rate of expansion and rate of shear of the perfect fluid. Since the perfect fluid is comoving with the normal congruence, these are the only kinematically relevant quantities. The associated evolution equations are:
\begin{align}
\dot H &= -H^2 -\frac{1}{3} \sigma_{ab} \sigma^{ab} -\frac{1}{3} x_0^2 + \frac{1}{6} (2-3\gamma) \rho \,, \label{eq:H-ev} \\
\dot \sigma_{ab} &= -3H\sigma_{ab} + 2 \epsilon^{ij}_{\;\;\; (a} \sigma^{}_{b)i} \Omega_j - {^3} S_{ab} + \pi_{ab} \,, \label{eq:shear-ev}
\end{align}
where ${^3} S_{ab}$ is the three-dimensional traceless Ricci tensor and $\pi_{ab}$ is the anisotropic stress tensor.

\subsection{Gauge-fixing \label{sub:gauge}}
We now switch to expansion normalized variables defined in appendix \ref{app:normalized}. In these variables the gauge fixing conditions introduced above can be summarized by 
\begin{align}
R_i &\equiv \frac{\Omega_i}{H} = \sqrt{3} ( \Sigma_\times,  \Sigma_3, - \Sigma_2 ) \,, \label{gauge:R} \\
A_i &= (A,0,0) \,, \label{gauge:A} \\
X_i &= (V_1, 0, V_3) \,. \label{gauge:V}
\end{align}  
The goal of this section is to show that one can further set
\be
\Sigma_2=0 \,,
\label{gauge:S2}
\ee
without loss of generality. In order to prove this we need the evolution equation for $\Sigma_2$:
\be
\Sigma_2' = \Sigma_2 (q-2-3\Sigma_+-\sqrt{3}\Sigma_-) \,, \label{S2}
\ee
where $q$ is the deceleration parameter defined in Appendix \ref{app:normalized} and $'$ denotes differentiation with respect to dimensionless time defined below in (\ref{time:dimensionless}). It follows from (\ref{S2}) that the conditions $\Sigma_2>0$, $\Sigma_2=0$ and $\Sigma_2<0$ define invariant sets. In particular, the condition $\Sigma_2=0$ is preserved in time. However, since the conditions (\ref{gauge:R})-(\ref{gauge:V}) fix the gauge completely in class B models with $V_3\neq 0$, we cannot set $\Sigma_2=0$ as an initial condition simply by a choice of frame. The argument relies on some constraint equations. We need the $(0,2)$ component of the Einstein equation \cite{bok:EllisWainwright}: 
\be
0=N_+ \Sigma_3 + 3A\Sigma_2 + \sqrt{3} \Sigma_2 N_\times - \sqrt{3} \Sigma_3 N_- \,.
\label{equation:0i}
\ee
We also need the $(\mu, \nu)=(i, j)$ components of the ``Maxwell'' equation (\ref{eqm:dF}). This gives two equations, which are expressed in expansion normalized form in equations (\ref{C1}) and (\ref{C2}) in the subsection below.\footnote{In section \ref{sub:evo} the gauge condition (\ref{gauge:S2}) is already implemented, but there is no appearance of shear variables in the space-space components of equation (\ref{eqm:dF}).} There are two ways to satisfy (\ref{C1}) and (\ref{C2}). The first one is by the condition $V_3=0$. In that case $V_i$ is parallel to $A_i$ and we have an unused rotation around the frame vector $\mathbf e_1$ that can be used to set the initial condition $\Sigma_2=0$. The second one is by the condition $(A,\, N_+) = \sqrt{3}(N_\times, \, N_-)$. In that case constraint equation (\ref{equation:0i}) directly gives $\Sigma_2=0$ or $N_\times=0$. If $N_\times = 0$ it follows that $A=0$ and that the only non-vanishing component of the matrix $N_{ab}$ is $N_{22}=2\sqrt{3}N_-$. This corresponds to Bianchi type II if $N_-\neq 0$ and Bianchi type I if $N_- = 0$. In this case $N_{ab}$ is invariant under a rotation of the frame with respect to the frame vector $\mathbf e_2$. The gauge condition $V_2=0$ is also invariant under the same rotation. Since the object $[\Sigma_2, \Sigma_\times]^T$ transforms as a spin-1 object under a frame rotation with respect to $\mathbf e_2$, we have freedom to choose the initial condition $\Sigma_2=0$, on top of the gauge conditions given by (\ref{gauge:R})-(\ref{gauge:V}).

To conclude, we have proved that, on top of the gauge fixing conditions (\ref{gauge:R})-(\ref{gauge:V}), the condition $\Sigma_2=0$ either follows directly from constraint equations (as in class B models with $V_3\neq0$) or can be chosen by an appropriate initial orientation of the frame. This is equivalent to the following statement:
\begin{lemma}
	In the gauge defined by (\ref{gauge:R})-(\ref{gauge:A}) the objects $[V_2, V_3]^T$ and $[\Sigma_2, \Sigma_3]^T$ can be chosen parallel (as column vectors in $\mathbb{R}^2$) without restricting the space of physical models. 
	\label{Lemma:gauge}
\end{lemma}
\begin{proof}
	Arguments above in conjuncture with the transformation properties in appendix \ref{app:transformations}.
\end{proof}
We have chosen the particular orientation of the orthonormal frame in which $V_2$ and $\Sigma_2$ vanish simultaneously. This fixes the gauge uniquely, for the most general Bianchi models, up to some discrete transformations discussed in Appendix \ref{app:discrete}.

\subsection{Evolution and constraint equations \label{sub:evo}} 
Here we write down the constrained autonomous system that corresponds to the Einstein and matter equations in expansion normalized variables. Later this set of evolution equations and constraint equations, will be referred to, collectively, as the \emph{dynamical system}. The system assumes the gauge fixing conditions (\ref{gauge:R})-(\ref{gauge:S2}) developed above and agree with the equations presented in section 5 of \cite{Ben18} by specializing to this gauge. In the following two subsections some results regarding orthogonal models and diagonalizability, that follow directly, are presented.

\paragraph{Evolution equations} 
\begin{align}
\Omega' &= \Omega \left( 2q+2-3\gamma \right) \,, \label{O} \\
\Theta' &= \Theta (q-2) -2AV_1 \,, \label{T} \\ 
V_1' &= V_1(q+2\Sigma_+)-2\sqrt{3} \Sigma_3 V_3 \,, \label{V1} \\
V_3' &= V_3(q-\Sigma_+ + \sqrt{3} \Sigma_- ) \,, \label{V3} \\
\Sigma_+' &= \Sigma_+ (q-2)+3\Sigma_3^2 - 2 (N_-^2+N_\times^2) + V_3^2 -2V_1^2 \,, \label{Sp} \\
\Sigma_-' &= \Sigma_- (q-2)+\sqrt{3}( 2\Sigma_\times^2 - \Sigma_3^2) + 2 A N_\times - 2N_-N_+ - \sqrt{3} V_3^2 \,, \label{Sm} \\
\Sigma_\times' &= \Sigma_\times (q-2-2\sqrt{3} \Sigma_-) - 2AN_- - 2 N_+ N_\times \,, \label{Sc} \\
\Sigma_3' &= \Sigma_3 (q-2-3\Sigma_++\sqrt{3}\Sigma_-)+2\sqrt{3}V_1V_3 \,, \label{S3}  \\
A' &= A (q+2\Sigma_+) \,, \label{A} \\ 
N_+' &= N_+ (q+2\Sigma_+) + 6\Sigma_-N_- + 6 \Sigma_\times N_\times \,, \\
N_-' &= N_- (q+2\Sigma_+) + 2\sqrt{3}\Sigma_\times N_\times + 2 \Sigma_- N_+ \,, \\
N_\times' &= N_\times (q+2\Sigma_+) - 2\sqrt{3}\Sigma_\times N_- + 2 \Sigma_\times N_+ \,, \label{Nx} 
\end{align}
where 
\be
q = \left(\tfrac{3}{2}\gamma -1\right)\Omega + 2\Theta^2 + 2\Sigma^2  , \quad \Sigma^2 \equiv \Sigma_+^2 +\Sigma_-^2 +\Sigma_\times^2 + \Sigma_3^2 \,.
\label{def:sigma}
\ee
Here $'$ denotes a derivative with respect to the dimensionless time $\tau$ that is defined via the Hubble scalar: 
\be
H = \frac{d\tau}{dt} \,.
\label{time:dimensionless}
\ee
The evolution equation (\ref{T}) for the temporal field strength component $\Theta$ follows from (\ref{eqm:dsF}), whereas the evolution equations (\ref{V1}) and (\ref{V3}) for the spatial components $V_1$ and $V_3$ is (\ref{eq:x}) rewritten in expansion normalized variables. The shear evolution equations (\ref{Sp})-(\ref{S3}) are the trace-free space-space components of the Einstein equation, and the evolution equations (\ref{A})-(\ref{Nx}) for the geometrical variables follow from the Jacobi identities for the frame basis vectors; see \cite{Ben18} and \cite{bok:EllisWainwright} for details.

\paragraph{Constraint equations} The variables are subject to multiple constraints: 

\begin{align}
C_1 &=0, \quad C_1 \equiv V_3 (A-\sqrt{3}N_\times) \,, \label{C1} \\
C_2 &=0, \quad C_2 \equiv V_3 (N_+-\sqrt{3}N_-) \,, \label{C2} \\
C_3 &=0, \quad C_3 \equiv \Sigma_+ A + \Sigma_- N_\times - \Sigma_\times N_- - \Theta V_1 \,, \label{C3} \\
C_4 &=0, \quad C_4 \equiv \Sigma_3(N_+  - \sqrt{3} N_-) \,, \label{C4} \\
C_5 &=0, \quad C_5 \equiv  \Sigma_3 (N_\times - \sqrt{3}A )  - 2 \Theta V_3 \,, \label{C5} \\
C_6 &=0, \quad C_6 \equiv -1+\Omega + \Theta^2 + V_1^2 + V_3^2  + \Sigma^2 + A^2  + N_-^2 + N_\times^2 \,. \label{C6} 
\end{align}
Constraints (\ref{C1})-(\ref{C2}) follow from the space-space components of the ``Maxwell equation'' (\ref{eqm:dF}). The other constraints follow from the Einstein equation. Specifically, (\ref{C3}), (\ref{C4}) and(\ref{C5}) correspond to (0,1), (0,2) and (0,3) equations of the Einstein equation, respectively. 

\subsubsection{Orthogonal models}
By assumption the perfect fluid is \emph{non-tilted}, in the sense that its four-velocity is identified with the normal congruence $u^\mu$. For the scalar field, on the other hand, our only physical assumption is homogeneity of the gradient: $X_\mu=X_\mu(t)$. Analogously, we say that the scalar field is \emph{non-tilted} in the case that its energy flux vanishes relative to the normal congruence. This is only the case if $\Theta V_i=0$ for $i=1,2,3$, see \cite{Mik18} for the covariant decomposition of the energy-momentum tensor explicitly written down. Otherwise, we say that the scalar field is \emph{tilted}. A cosmological model is referred to as \emph{orthogonal} if \emph{all} matter types are non-tilted. This requires that the conditions  $\Theta V_1=0$ and $\Theta V_3=0$ are preserved in time. By inspecting the evolution equations we find three types of orthogonal models containing the scalar field:

\begin{enumerate}
	\item All models with $\nabla_\mu \varphi$ parallel to $u^\mu$:
	\be
	\Theta^2>0\,, \quad V_1=V_3=0 \,.
	\ee
	\item Class A models with $\nabla_\mu \varphi$ orthogonal to $u^\mu$: \newline
	\be
	\Theta = 0\,, \quad V_1^2+V_3^2>0\,, \quad A=0 \,.
	\ee 
	\item Class B models with: \newline
	\be
	\Theta = V_1=0\,, \quad V_3^2>0 \,, \quad \Sigma_3=0\,, \quad A>0 \,.
	\ee 
\end{enumerate}

\subsubsection{Diagonalizability}
It is well-known that $\Sigma_{ij}$ and $N_{ij}$ are simultaneously diagonalizable in spacetimes of Bianchi class A containing a non-tilted perfect fluid \cite{main:Ellis1969,Wainwright-Hsu:1989}. It is interesting to see how this result transfers to the considered model defined by the evolution equations and constraint equations above. We find the result to hold even in the presence of an anisotropic field strength one-form given that the model is orthogonal:  
\begin{lemma}
	$N_{ij}$ and $\Sigma_{ij}$ are simultaneously diagonalizable in orthogonal models of Bianchi class A (Bianchi type II, VI$_0$ and VII$_0$).  
	\label{Lemma:dia}
\end{lemma}
\begin{proof}
	For orthogonal class A models we have $A=0$, $\Theta V_1=0$ and $\Theta V_3=0$. Note from the evolution equations that these three conditions are preserved in time and thus define a cosmological model. In that case constraint (\ref{C3}) gives
	\be
	\Sigma_- N_\times = \Sigma_\times N_- \,.
	\label{dia:useful}
	\ee 
	First assume $V_3=\Sigma_3=0$. A frame rotation around $\mathbf e_1$ is then compatible with the gauge conditions $\Sigma_2=V_2=0$. This can be used to diagonalize $N_{ab}$, i.e. setting $N_\times=0$.  It follows from (\ref{dia:useful}) that $\Sigma_\times$ vanishes (and hence $\Sigma_{ab}$ is diagonal) if $N_- \neq 0$, or $\Sigma_\times$ can be chosen to vanish if $N_-=0$ (since in that case $N_{ab}$ is invariant under  frame rotations around $\mathbf e_1$). Thus $N_{ij}$ and $\Sigma_{ij}$ are simultaneously diagonalizable under the assumption $V_3=\Sigma_3=0$. Next consider the case that one or both of $V_3$ and $\Sigma_3$ are non-zero. It follows from the constraint equations that $N_\times=0$ and $N_+=\sqrt{3}N_-$ so that $N_{ab}$, $\Sigma_{ab}$ and $V_a$ have matrix representations on the form (\ref{matrix-representation}). In this case frame rotations around $\mathbf e_2$ are compatible with the gauge conditions $\Sigma_2=V_2=0$ and allow simultaneous diagonalization of $N_{ab}$ and $\Sigma_{ab}$.  	 
\end{proof}

\subsection{Cauchy problem \label{sub:cauchy}} 
General relativity has a well-posed Cauchy problem, which here implies that if all constraint equations (\ref{C1})-(\ref{C6}) are satisfied at a given instant of time, they are satisfied at any later (or earlier) time. To prove that the constraint and evolution equations above are mutually consistent in this way, provides a powerful check of the equations. We first assume that the functions $C_1, C_2, \dots, C_6$ defined above are time dependent. Their evolution equations are then obtained by differentiating with respect to $\tau$ and using the evolution equations above: 
\begin{align}
C_1'&= C_1 (2q+\Sigma_++\sqrt{3}\Sigma_-)-2\sqrt{3} C_2 \Sigma_\times  \,, \\
C_2'&= C_2 (2q+\Sigma_+-\sqrt{3} \Sigma_-)\,, \\
C_3'&= 2C_3(-1+q+\Sigma_+) + C_1 V_3  - \sqrt{3} C_5 \Sigma_3 \,, \\
C_4'&= C_4 (-2+2q-\Sigma_+ -\sqrt{3}\Sigma_-) + 2\sqrt{3} C_2 V_1 \,, \\
C_5'&= C_5 (-2+2q-\Sigma_+ +\sqrt{3}\Sigma_-) + 2 (C_4\Sigma_\times-C_1 V_1) \,, \\
C_6'&= 2q C_6 + 4A C_3\,. 
\end{align}
It follows from these equations that if all constraints are satisfied initially 
\be
C_1(\tau_0)=C_2(\tau_0)=\dots=C_6(\tau_0)=0 \,,
\ee
then they are always satisfied. This shows that the dynamical system under consideration has a well-posed Cauchy problem, as it should have.

\section{State space and invariant sets \label{ch:invariant}}
We start this section by identifying some discrete symmetries in \ref{sub:discrete} and establishing the classification in Bianchi types in \ref{ch:Bianchi}. This prepares section \ref{ch:BIS} where state space is written as a union of disjoint invariant sets, each of which is identified as a particular cosmological model. The main results are summarized in table \ref{tab:sets:C} and \ref{tab:sets:D}. Our routine for counting degrees of freedom, used throughout the paper, is established in section \ref{sub:counting}.

\subsection{Discrete symmetries \label{sub:discrete}}
The dynamical system enjoys some discrete symmetries that will simplify the analysis. Specifically, all evolution equations (\ref{O})-(\ref{Nx}) \emph{and} constraint equations (\ref{C1})-(\ref{C6}) are invariant under each of the following transformations: 
\begin{align}
(\Theta, V_1, V_3) & \quad\rightarrow\quad (-\Theta, -V_1, -V_3) \,, \label{inv1} \\
(A, N_+, N_-, N_\times) & \quad\rightarrow\quad (-A, -N_+, -N_-, -N_\times) \,, \label{inv2} \\
(V_3, \Sigma_3) & \quad\rightarrow\quad (-V_3, -\Sigma_3) \,, \label{inv3} \\
(V_1, V_3, \Sigma_\times, A, N_\times) & \quad\rightarrow\quad (-V_1, -V_3, -\Sigma_\times,  -A, -N_\times) \,. \label{inv4}  
\end{align}
The first symmetry (\ref{inv1}) can be traced back to the invariance of the ``Maxwell equations'' and the energy-momentum tensor, under $\boldsymbol{\mathcal F} \rightarrow -\boldsymbol{\mathcal F}$. The second symmetry (\ref{inv2}) can be traced back to the invariance of the metric (\ref{line-element}) under a parity flip $(\boldsymbol{\rm d}t,  \boldsymbol{\omega}^a) \rightarrow (\boldsymbol{\rm d}t,  -\boldsymbol{\omega}^a)$.\footnote{This follows by noting from (\ref{strcoeff}) that (\ref{inv2}) corresponds to flipping signs on the spatial structure coefficients, $\tensor{\gamma}{^m_{ab}} \rightarrow - \tensor{\gamma}{^m_{ab}}$, which corresponds to $(\boldsymbol{\rm d}t,  \boldsymbol{\omega}^a) \rightarrow (\boldsymbol{\rm d}t,  -\boldsymbol{\omega}^a)$ according to (\ref{dw}).} At last, the gauge fixing conditions (\ref{gauge:R})-(\ref{gauge:S2}) leave some discrete frame rotations identified in appendix \ref{app:discrete}. The symmetry (\ref{inv3}) and (\ref{inv4}) correspond to spatial half-turns of the frame around $\mathbf e_1$ and $\mathbf e_2$, respectively. 

The relevance of these symmetries is that they imply equivalence under evolution between state vectors related by any of the transformations above. To be concrete, let us write $X\sim Y$ if $X$ and $Y$ are two state vectors identical up to (\ref{inv1}), (\ref{inv2}), (\ref{inv3})  and/or (\ref{inv4}). It follows from the invariance of the evolution equations that if $X_0\sim Y_0$ for two sets of initial conditions, then the same is true at any time:
\be
X_0\sim Y_0 \implies X(\tau) \sim Y(\tau) \,.
\ee
Thus two state vectors related by $X\sim Y$ evolve equivalently. We will next employ this observation to introduce some sign conventions.

Since the signs of $V_3$ and $A$ are preserved by the evolution equations (\ref{V3}) and (\ref{A}), respectively, state space can be restricted to non-negative values of these variables 
\be
V_3 \ge 0, \quad A \ge 0 \,,
\label{gauge:lastfix}
\ee
without loss of generality. This sign convention will be used throughout the paper.\footnote{Note that on top of the this convention, (\ref{inv1})-(\ref{inv4}) still allow flipping the sign of, say, $\Sigma_3$. However, the sign of $\Sigma_3$ is generally not preserved in time as seen from evolution equation (\ref{S3}).}

\subsection{Bianchi classification \label{ch:Bianchi}}
The Bianchi classification is based on the structure of the covector $A_i$ and the matrix $N_{ab}$. The evolution equation (\ref{A}) for $A$ implies that $A=0$ and $A>0$ define invariant subspaces. These are the Class A and Class B models, respectively. These classes are further divided into Bianchi types based on the three eigenvalues of the matrix $N_{ab}$, which are:
\be
N_1=0 \,, \quad N_2 = N_+ + \sqrt{3N_-^2+3N_\times^2} \,, \quad N_3 = N_+ - \sqrt{3N_-^2+3N_\times^2} \,. 
\ee
In the orthonormal frame approach the eigenvalues are constants on the group orbits, but evolve in time. Therefore the invariant classification is actually based on the sign of the eigenvalues. To see that the sign of each eigenvalue is preserved in time, consider the evolution equations for $N_+$, $N_-$ and $N_\times$ from which it follows that  
\be
F' = 2(q+2\Sigma_+)F\,, \quad F \equiv N_2 N_3 = N_+^2-3N_-^2-3N_\times^2 \,.
\label{eq:F}
\ee
It follows from this equation, which is independent of the gauge choice (\ref{gauge:R}), that $F>0$, $F=0$ and $F<0$ are invariant subspaces. Consequently, the signs of $N_2$ and $N_3$ are preserved in time. 

By comparing the evolution equations for $F$ (\ref{eq:F}) with the evolution equation for $A$ (\ref{A}) we note that $F$ and $A^2$ evolve proportionally. Thus, in models with non-zero $F$ we have a constant of motion $h$ defined through
\be
A^2 = h (N_+^2 - 3N_-^2 - 3N_\times^2) \,.
\label{def:h}
\ee
This is the Bianchi type VI$_h$ if $h<0$ $(F<0)$ and Bianchi type VII$_h$ if $h>0$ ($F>0$). 

For convenience we list the definitions of Bianchi types I-VII$_h$ in terms of $A$ and the eigenvalues of $N_{ab}$, and translate them to our considered variables:\footnote{For $F=0$ there are two different possibilities for $N_{ab}$: it has either three zero eigenvalues ($N_1=N_2=N_3=0$) or one non-zero eigenvalue. This corresponds to type I and II, respectively, in class A models, and to type V and and IV, respectively in Class B models.}
\begin{enumerate}
	\item Type I: class A with three zero eigenvalues ($F=0$), \newline 
	\be
	A = 0, \quad  N_+ = N_-= N_\times = 0 \,.
	\label{type:I}
	\ee
	\item Type II: class A with one non-zero eigenvalue ($F=0$),
	\be
	A=0, \quad N_+^2 = 3 (N_-^2 + N_\times^2) > 0 \,.
	\label{type:II}
	\ee
	\item Type III: class B with two non-zero eigenvalues with opposite sign ($F<0$), it is the particular case $h=-1$ of type VI$_h$ (III$=$VI$_{-1}$),
	\be 
	A^2 = 3N_-^2+3N_\times^2-N_+^2 > 0 \,. \label{type:III}
	\ee
	\item Type IV: class B with one non-zero eigenvalue ($F=0$),
	\be
	 A>0, \quad N_+^2 = 3 (N_-^2 + N_\times^2) > 0 \,.
	\ee
	\item Type V: class B with three zero eigenvalues ($F=0$), \newline 
	\be
	A > 0, \quad  N_+ = N_-= N_\times = 0 \,. \label{type:V}
	\ee
	\item Type VI$_0$: class A with two non-zero eigenvalues with opposite sign ($F<0$),
	\be
	A=0, \quad N_+^2-3N_-^2-3N_\times^2 < 0 \,.
	\ee
	\item Type VI$_h$ ($-1 \neq h <0$):  class B with two non-zero eigenvalues with opposite sign ($F<0$),
	\be
	A^2=h(N_+^2-3N_-^2-3N_\times^2)>0 \,.
	\ee 
	\item Type VII$_0$: class A with two non-zero eigenvalues with the same sign ($F>0$),
	\be
	A=0, \quad N_+^2-3N_-^2-3N_\times^2 > 0 \,.
	\ee 
	\item Type VII$_h$ ($h>0$):  class B with two non-zero eigenvalues with the same sign ($F>0$), 
	\be 
	A^2=h(N_+^2-3N_-^2-3N_\times^2)>0 \,.
	\label{type:VIIh}
	\ee
\end{enumerate}

\subsection{Bianchi invariant sets \label{ch:BIS}}

\subsubsection{State space \label{sub:statespace}}
The state vector is 
\be
\mathrm{X} = (\Theta, V_1, V_3, \Sigma_+, \Sigma_-, \Sigma_\times, \Sigma_3, A, N_+, N_-, N_\times) \,,
\ee
modulo the five constraint equations (\ref{C1}), (\ref{C2}), (\ref{C3}), (\ref{C4}) and (\ref{C5}) and the inequality
\be
\Theta^2 + V_1^2 + V_3^2  + \Sigma^2 + A^2  + N_-^2 + N_\times^2 \le 1 \,.
\label{inequality}
\ee
Here the Hamiltonian constraint (\ref{C6}) has been used to eliminate $\Omega$ and the inequality corresponds to positive energy $\Omega\ge0$.  Let $D$ denote \emph{state space}, i.e. the set of all state vectors $\mathrm X$ in $\mathbb{R}^{11}$ that satisfy the five constraints and the inequality:
\be
D = \left\{ \;  \mathrm X \in \mathbb{R}^{11} \; | \; C_1=C_2=C_3=C_4=C_5=0\,,\; (\ref{inequality}) \; \right\} .
\label{D:def}
\ee

As observed in section (\ref{ch:Bianchi}) the signs of $A$ and the eigenvalues of $N_{ab}$ are preserved not only on the group orbits, but also in time. The Bianchi type of a given state vector $\mathrm X$ is therefore invariant under time evolution. Let $\mathcal B(\#)$ denote the union of all orbits in $D$ corresponding to a specific Bianchi type $\#$, where $\#$ denotes one of the Bianchi types I-VII$_h$ listed in section \ref{ch:Bianchi}. For instance $\mathcal B(\text{I})$ is the union of all orbits in $D$ of Bianchi type I. Following a standard language \cite{bok:EllisWainwright}, these sets will be referred to as \emph{Bianchi invariant sets}.

The conditions that define each invariant set is obtained by combining constraint equations (\ref{C1})-(\ref{C5}) with the defining conditions (\ref{def:h})-(\ref{type:VIIh}) for each Bianchi type. By inspection we note that the ``easiest'' way to satisfy all constraint equations is by choosing the conditions $\Sigma_3=V_3=0$. The associated invariant sets, of a specific Bianchi type $\#$, will be written as
\be
\mathcal C(\#) = \mathcal C^+(\#) \cup \mathcal C^0(\#) \,,
\ee
where the disjoint subsets are defined by the following conditions:
\begin{align}
\mathcal C^+(\#):& \qquad V_1^2>0\,,\quad V_3=0\,, \quad \Sigma_3=0\,, \quad (\ref{C3})\,, \qquad \text{``imperfect branch''}\,, \label{def:C+} \\
\mathcal C^0(\#):& \qquad V_1=0\,,\quad\; V_3=0\,, \quad \Sigma_3=0\,, \quad (\ref{C3})\,, \qquad \text{``perfect branch''}\,, \label{def:C0}
\end{align}
on top of the defining condition for the specific Bianchi type $\#$, see the enumerated list in section \ref{ch:Bianchi}. Note that the superscript indicates if imperfect matter is present (+) or absent (0). The corresponding sets will be referred to as the \emph{imperfect branch} and \emph{perfect branch}, respectively, which are invariant sets themselves. In Bianchi types IV, V, VI$_h$ with $h\neq \{ -1/9,\,-1 \}$ and VII$_h$ with $h\ge 0$ these sets are ``complete'', in the sense that they include all state vectors in $D$ of the given Bianchi type. This results from constraint equations (\ref{C1})-(\ref{C5}) by considering the complement of $V_3=\Sigma_3=0$. If $V_3=0$ and $\Sigma_3\neq0$ the constraints give $(N_+, N_\times)=\sqrt{3}(N_-, A)$, which implies Bianchi type I, II or VI$_{-1/9}$. If $V_3\neq0$ the constraints give $(N_+, A)=\sqrt{3}(N_-, N_\times)$, which implies Bianchi type I, II or III.

For Bianchi type I, II, III and VI$_{-1/9}$ the situation is more complicated and, as we shall see, the physics richer. These cases are investigated in subsections \ref{subsub:I} to \ref{subsub:exceptional}. Therein additional sets, with notation on the form
\be
\mathcal D(\#) = \mathcal D^+(\#) \cup \mathcal D^0(\#) \,, 
\ee
are defined so that state space $D$ can be build from a union of disjoint invariant sets. Again, the superscript indicates if the matter sector contains the isotropy violating field strength one-form:
\begin{align}
\mathcal D^+(\#):& \qquad  V_1^2+V_3^2 > 0 \,, \quad \text{``imperfect branch''} \,,  \\
\mathcal D^0(\#):& \qquad V_1^2+V_3^2 = 0 \,,  \quad \text{``perfect branch''} \,.
\end{align}

State space $D$ can now be written as a union of \emph{disjoint} invariant sets in the following way:
\be 
D = \bigcup_{i=1}^{10} \mathcal B (\#_i) =  \mathcal B (\text{I}) \cup \mathcal B (\text{II}) \cup \dots \cup \mathcal B (\text{VII}_0) \cup \mathcal B (\text{VII}_h) \,,
\label{D:subspaces}
\ee 
where
\begin{align}
\mathcal B(\text{I}) &= \mathcal D^+(\text{I}) \cup \mathcal D^0(\text{I})\,, \label{BIS:I} \\
\mathcal B(\text{II}) &= \mathcal{C^+(\text{II})} \cup \mathcal C^0(\text{II})   \cup \mathcal D^+(\text{II})|_{V_3=0}  \cup \mathcal D^+(\text{II})|_{V_3>0} \cup \mathcal D^0(\text{II}) \,, \label{BIS:II} \\
\mathcal B(\text{III}) &= \mathcal C^+(\text{III}) \cup \mathcal C^0(\text{III}) \cup \mathcal D^+(\text{III}) \,, \label{BIS:III} \\
\mathcal B(\text{IV}) &= \mathcal C^+(\text{IV}) \cup \mathcal C^0(\text{IV}) \,, \label{BIS:IV} \\
\mathcal B(\text{V}) &= \mathcal C^+(\text{V}) \cup \mathcal C^0(\text{V}) \,, \label{BIS:V} \\
\mathcal B(\text{VI}_0) &= \mathcal C^+(\text{VI}_0) \cup \mathcal C^0(\text{VI}_0) \,, \label{BIS:VI0} \\
\mathcal B(\text{VI}_{-1/9}) &= \mathcal C^+(\text{VI}_{-1/9}) \cup \mathcal C^0(\text{VI}_{-1/9}) \cup \mathcal D^+(\text{VI}_{-1/9}) \cup \mathcal D^0(\text{VI}_{-1/9}) \,, \label{BIS:VIexc} \\
\mathcal B(\text{VI}_h) &= \mathcal C^+(\text{VI}_h) \cup \mathcal C^0(\text{VI}_h) \,, \text{ where } 0>h\neq \{-1/9, -1\} \,, \label{BIS:VIh} \\
\mathcal B(\text{VII}_0) &= \mathcal C^+(\text{VII}_0) \cup \mathcal C^0(\text{VII}_0) \,, \label{BIS:VII0} \\ 
\mathcal B(\text{VII}_h) &=  \mathcal C^+(\text{VII}_h) \cup \mathcal C^0(\text{VII}_h) \,, \text{ where } h>0 \,.  \label{BIS:VIIh}
\end{align}

Tables \ref{tab:sets:C} and \ref{tab:sets:D} summarize main properties of the invariant sets above. In the second column the corresponding cosmological model is given: 
\begin{itemize}
	\item 'perfect' indicates that imperfect matter is absent ($V_i=\vec 0$), so that \emph{only} perfect matter ($\Omega$ and $X_0=\Theta$) is present,
	\item  'imperfect' indicates that the imperfect matter is present ($V_i\neq\vec 0$),
	\item 'orthogonal' indicates that the cosmological model is orthogonal, as defined in section \ref{sub:matter}.\footnote{Note that all perfect models are orthogonal, whereas the converse is not always true:  $\mathcal D^+(\text{II})|_{V_3>0}$ is imperfect, yet orthogonal.}
\end{itemize}
 
We also count the degrees of freedom associated with each set. Our routine is given in the subsection below.  The number of independent modes, $d$, is given in the tables together with the number of redundant gauge modes, $r$. The former is a gauge invariant quantity which measures the generality of the model and is similar to the dimension of the associated dynamical system. The latter is a non-fundamental quantity which, if non-zero, reflects that the gauge is not fixed completely by the gauge fixing conditions (\ref{gauge:R})-(\ref{gauge:S2}). See subsection \ref{sub:counting} for details. 

Physically, a distinction between the imperfect sets of type $\mathcal C^+(\#)$ and $\mathcal D^+(\#)$, is that it is only in the latter that the direction of the spatial field-strength one-form $(X_i)$ may rotate relative to a gyroscope. Consider a gyroscope initially oriented so that the spin axis is aligned with $X_i$. In all sets of type $\mathcal C^+(\#)$ the spin axis remains aligned with $X_i$. This follows by noting that the spatial direction of the field strength one-form $X_i=(V_1, 0 ,0)$ is parallel to the axis of the frame rotation $R_i=(\sqrt{3}\Sigma_\times, 0, 0)$, which is unique for the sets $\mathcal C^+(\#)$. In the sets $\mathcal D^+(\#)$ the direction of $X_i(t)$ will generally move away from the spin axis of the gyroscope, and in this sense $X_i$ is rotating.

\subsubsection{Space of initial data \label{sub:counting}}

For spatially homogeneous models the dimension of the space of initial data provides a measure of the generality of each model. This measure is reliable only because the functions involved in the counting are constants \cite{bok:EllisWainwright}. In each invariant set the dimension of the space of initial data, $d$, is the number of \emph{independent} parameters required to define initial conditions on a Cauchy surface ($\tau=\text{constant}$). This is the same as the dimension of the dynamical system associated with each invariant set (upon fixing the gauge completely). 

To obtain the dimension $d$ of a given set one must look at the defining conditions for the set and count the number of constraint equations ($c$) and the number of (redundant) gauge modes $r$. The latter requires a definition. If two state vectors are identical upon a frame rotation, $Y=T(X)$ for some transformation $T$, we say that the state vectors are equivalent and write $X\sim Y$. The \emph{gauge orbit} of a state vector $X$ (in a set $\mathcal B$) is the set of all state vectors $Y$ (in $\mathcal B$) that satisfy $Y\sim X$. We then define $r$ as the dimension of the gauge orbit.  We can then obtain the dimension $d$, algorithmically, as the number of coordinates in $\mathbb R^{11}$ minus the number of constraint equations ($c$) minus the dimension of the gauge orbit ($r$):
\be
d=11-c-r \,.
\label{counting:alg}
\ee

Informally, $d$ is the number of independent modes associated with \ $\{\rho, x_\mu, \sigma_{ab},a_i, n_{ab}\}$, whereas $r$ is the number of redundant gauge modes left by the gauge-fixing conditions (\ref{gauge:R})-(\ref{gauge:S2}). Note that $r$ is not necessarily the dimension of the rotation group $T$, since in sets with higher symmetries the state vector is invariant under a 1 or 3 dimensional group of rotations. Our definition of $r$, as the dimension of the gauge orbit\footnote{The term 'gauge orbit' is a bit misleading in the case of Bianchi type I, since in this case the set of transformations that preserve $\Sigma_2=V_2=0$ do not form a group.  However, $r$ is still the number of redundant modes and the algorithm (\ref{counting:alg}) is still applicable, see section \ref{subsub:I} for details.}, ensures that $d$ has the same meaning in models with higher symmetries, as considered in section \ref{ch:higher}. In such models the absence of a unique orientation of the frame does not introduce any dependency among the variables.  As for the number of constraints $c$, note that only constraint \emph{equations} are counted. Constraint \emph{inequalities} are not counted since they constrain the volume of state space and not its dimension.\footnote{As an example consider the set $\mathcal D^0(\text{VI}_{-1/9})$ that is defined by the conditions (\ref{def:D0VIexp}), and for which table \ref{tab:sets:D} states $c=5$. The five constraints counted are: $V_1=0$, $V_3=0$, $N_\times=\sqrt{3}A$, $N_+ = \sqrt{3} N_-$ and (\ref{C3}). Note that (\ref{def:h}) with $h=-1/9$ is satisfied trivially. In the set $\mathcal C^0(\text{VI}_{-1/9})$, on the other hand, (\ref{def:h}) must be counted in addition to $V_1=0$, $V_3=0$, $\Sigma_3=0$ and $(\ref{C3})$. Thus $c=5$ also for $\mathcal C^0(\text{VI}_{-1/9})$. \label{fn:ex:c}}

In tables \ref{tab:sets:C} and \ref{tab:sets:D} the dimension $d$ of each invariant set is given. In Bianchi type VI$_h$ and VII$_h$ the group parameter $h$ is considered fixed.  In the perfect branch the counting agrees with Wainwright and Hsu \cite{Wainwright-Hsu:1989} and Hewitt and Wainwright \cite{Hewitt-Wainwright:1993} who considered orthogonal Bianchi class A and B models, respectively, containing a perfect fluid. Also see table 9.5 in \cite{bok:EllisWainwright} and table 18.3 in \cite{bok:EllisMaartensMacCallum}. Upon comparison it must be taken into account that in the perfect branch we actually have \emph{two} non-tilted perfect fluids: a stiff fluid ($\Theta$) and the usual $\gamma$-law perfect fluid ($\Omega$). For the perfect branch the dimension $d$ is therefore \emph{two} larger than in the corresponding vacuum model.

The number of redundant gauge modes $r$ are also summarized in the tables. In section \ref{sub:gauge} we already fixed the gauge as much as possible without restricting state space $D$.  Thus $r=0$ for the most general sets, which are found in models of Bianchi type III and VI$_{-1/9}$, both with dimension $d=7$. In all sets of type $\mathcal C(\#)$, summarized in table \ref{tab:sets:C}, the gauge is fixed only up to a rotation around $\mathbf e_1$ which gives a one-dimensional gauge orbit ($r=1$).  This is even the case in imperfect Bianchi models of class B ($A\neq0$), since the covectors $A_i$ and $V_i$ are aligned in all sets of type $\mathcal C(\#)$. For sets of type $\mathcal D(\#)$, summarized in table \ref{tab:sets:D}, the counting is a bit more complicated and details are given in subsections \ref{subsub:I} - \ref{subsub:exceptional}.

\begin{table}[H]
	\newcommand\T{\rule{0pt}{2.6ex}}
	\newcommand\B{\rule[-1.2ex]{0pt}{0pt}}
	\begin{center}
		\begin{tabular}{ l l l l l }
			\hline\hline 
			Notation \T\B   &  Cosmological model & $d$ &  $r$       \\ \hline
			\T\B $\mathcal C^+(\text{II})$ & Bianchi II, imperfect & 5  & 1    \\					
			\T\B $\mathcal C^0(\text{II})$ & Bianchi II, perfect, orthogonal & 4  & 1   \\			
			\T\B $\mathcal C^+(\text{III})$ & Bianchi III, imperfect & 6  & 1    \\					
			\T\B $\mathcal C^0(\text{III})$ & Bianchi III, perfect, orthogonal & 5  & 1   \\
			
			\T\B $\mathcal C^+(\text{IV})$ & Bianchi IV, imperfect & 6  & 1      \\					
			\T\B $\mathcal C^0(\text{IV})$ & Bianchi IV, perfect, orthogonal & 5  & 1   \\
			\T\B $\mathcal C^+(\text{V})$ & Bianchi V, imperfect & 4  & 1    \\					
			\T\B $\mathcal C^0(\text{V})$ & Bianchi V, perfect, orthogonal & 3  & 1  \\
			
			\T\B $\mathcal C^+(\text{VI}_h)$ & Bianchi VI$_h$ ($h\le 0$), imperfect & 6  & 1     \\					
			\T\B $\mathcal C^0(\text{VI}_h)$ & Bianchi VI$_h$ ($h\le 0$), perfect, orthogonal & 5  & 1 \\														
			\T\B $\mathcal C^+(\text{VII}_h)$ & Bianchi VII$_h$ ($h\ge 0$), imperfect & 6  & 1   \\					
			\T\B $\mathcal C^0(\text{VII}_h)$ & Bianchi VII$_h$ ($h\ge 0$), perfect, orthogonal & 5  & 1 \\														
			\hline\hline
		\end{tabular}
		\caption{ Invariant sets defined by (\ref{def:C+}) and (\ref{def:C0}) and properties of the corresponding cosmological model (see text for terminology and definitions). $d$ is the dimension of the space of initial data (gauge invariant), which is also the dimension of the corresponding dynamical system. $r$ is the number of unphysical gauge modes. In models of Bianchi type VI$_h$  and VII$_h$ the group parameter $h$ is considered fixed. } 		
		\label{tab:sets:C}
	\end{center}
\end{table}
\begin{table}[H]
	\newcommand\T{\rule{0pt}{2.6ex}}
	\newcommand\B{\rule[-1.2ex]{0pt}{0pt}}
	\begin{center}
		\begin{tabular}{ l l l l }
			\hline\hline 
			Notation \T\B   &  Cosmological model & $d$ &  $r$       \\ \hline
			\T\B $\mathcal D^+(\text{I})$ & Bianchi I, imperfect, orthogonal & 5 & 1   \\		
			\T\B $\mathcal D^0(\text{I})$ & Bianchi I, perfect, orthogonal & 3 & 2   \\		
			\T\B $\mathcal D^+(\text{II})|_{V_3=0}$ & Bianchi II, imperfect & 6  & 0  \\
			\T\B $\mathcal D^+(\text{II})|_{V_3>0}$ & Bianchi II, imperfect, orthogonal & 5  & 1 \\ 
			\T\B $\mathcal D^0(\text{II})$ & Bianchi II, perfect, orthogonal & 4  & 1 \\ 
			\T\B $\mathcal D^+(\text{III})$ & Bianchi III, imperfect & 7  & 0   \\					
			\T\B $\mathcal D^+(\text{VI}_{-1/9})$ & Bianchi VI$_{-1/9}$, imperfect & 7  & 0     \\					
			\T\B $\mathcal D^0(\text{VI}_{-1/9})$ & Bianchi VI$_{-1/9}$, perfect, orthogonal & 6  & 0   \\								
			\hline\hline
		\end{tabular}
		\caption{Invariant sets defined in subsections \ref{subsub:I} - \ref{subsub:exceptional}. See caption of table \ref{tab:sets:C}.} 
		\label{tab:sets:D}
	\end{center}
\end{table}

\subsubsection{Bianchi type I \label{subsub:I}}
All Bianchi type I orbits in state space belong to the imperfect set $\mathcal D^+(\text{I})$ or the perfect set $\mathcal D^0(\text{I})$. They are defined by the following conditions:
\begin{align}
\mathcal D^+(\text{I}):& \quad V_1^2+V_3^2 > 0 \,, \quad \Theta=0 \,, \quad (\ref{type:I}) \,,  \label{D+} \\
\mathcal D^0(\text{I}):& \quad  V_1^2+V_3^2 = 0 \,, \quad (\ref{type:I}) \,. \label{D0}
\end{align}
Note that $\mathcal C^+(\text{I})$ and $\mathcal C^0(\text{I})$ are subsets:
\be
\mathcal C^+(\text{I}) \subset \mathcal D^+(\text{I}) \,, \quad 
\mathcal C^0(\text{I}) \subset \mathcal D^0(\text{I}) \,,
\ee
and therefore left out from the union of disjoint sets on the right-hand side of (\ref{BIS:I}).

In order to calculate the dimension of initial data, we start by noting that $\mathcal D^+(\text{I})$ and $\mathcal D^0(\text{I})$ are defined by 5 and 6 constraint equations, respectively. Since Bianchi type I has flat spatial sections there is no unique tangent plane associated with the $G_2$ subgroup. We therefore have full freedom in choice of orientation of the spatial frame. In $\mathcal D^0(\text{I})$ one of the three gauge degrees of freedom has already been absorbed by the gauge choice $\Sigma_2=0$, so that $r=2$ redundant degrees of freedom remains; whereas in $\mathcal D^+(\text{I})$ two has been absorbed by the gauge conditions $V_2=\Sigma_2=0$, so that $r=1$. Using the algorithm (\ref{counting:alg}) the number of independent modes are: 
\begin{align}
\mathcal D^+(\text{I}):& \quad  d = 11 - 5 - 1 = 5 \,, \\
\mathcal D^0(\text{I}):& \quad  d = 11 - 6 - 2 = 3 \,. 
\end{align}
See \cite{Ben18} for examples on how the gauge can be fixed completely, resulting in dynamical systems with $d$ independent variables. 

It is informative to compare with the less general subsets:
\begin{align}
\mathcal C^+(\text{I}):& \quad  d = 11 - 7 - 1 = 3 \,, \\
\mathcal C^0(\text{I}):& \quad  d = 11 - 7 - 1 = 3 \,. 
\end{align}
As in all type $\mathcal C^+(\#)$ sets the direction of the spatial field strength one-form $X_i$ is fixed to the spin axis of a gyroscope. This explains why the number of physical modes ($d$) are similar in $\mathcal C^+(\text{I})$ and $\mathcal C^0(\text{I})$: a non-rotating ``vector'' ($V_i$) has no more physical modes than a ``scalar'' ($\Theta$). When the ``rotation modes'' of $V_i$ are released $d$ raises from 3 in $\mathcal C^+(\text{I})$ to 5 in $\mathcal D^+(\text{I})$. For more comments see \cite{Ben18}, where the rotation modes was shown to have an important role, qualitatively, in the dynamical system of Bianchi type I.

\subsubsection{Bianchi type II \label{subsub:II}} 
The set $\mathcal C(\text{II})$ does not cover all Bianchi type II orbits in state space $D$. We let $\mathcal D(\text{II})$ denote the \emph{remaining} orbits in $\mathcal B(\text{II})$: 
\be
\mathcal B(\text{II}) = \mathcal C(\text{II}) \cup \mathcal D(\text{II}) \,, \quad \mathcal C(\text{II}) \cap \mathcal D(\text{II}) = \emptyset \,.
\ee
$\mathcal D(\text{II})$ is the union of three disjoint sets
\be 
\mathcal D(\text{II}) =  \mathcal D^+(\text{II})|_{V_3>0} \cup \mathcal D^+(\text{II})|_{V_3=0} \cup \mathcal D^0(\text{II}) \,,
\ee
that are defined by the following conditions:
\begin{align}
&\mathcal D^+(\text{II})|_{V_3=0}: \qquad V_1^2>0\,, \quad V_3=0\,,\quad \Sigma_3\neq0\,,\quad \Sigma_\times N_- + \Theta V_1 = 0\,, \quad (\ref{common:DII}) \,, \\
&\mathcal D^+(\text{II})|_{V_3>0}: \qquad V_3>0\,,\quad \Theta = 0\,,\quad \Sigma_\times = 0 \,, \quad (\ref{common:DII}) \,, \label{def:DII:1} \\
&\mathcal D^0(\text{II}): \qquad \qquad \;  V_1^2 + V_3^2 = 0\,,\quad \Sigma_3 \neq 0 \,,\quad  \Sigma_\times = 0 \,, \quad (\ref{common:DII}) \,, \label{def:DII:3}
\end{align}
where the common conditions are
\be
A=0\,, \quad N_\times=0\,, \quad N_+ = \sqrt{3} N_- \neq 0 \,. 
\label{common:DII}
\ee

Upon counting constraints and redundant modes we obtain the number of independent modes using the algorithm (\ref{counting:alg}):
\begin{align}
&\mathcal C^+(\text{II}): \qquad \qquad \; d = 11-5-1 = 5 \,, \\
&\mathcal C^0(\text{II}): \qquad \qquad \;\, d = 11-6-1 = 4 \,, \\
&\mathcal D^+(\text{II})|_{V_3=0}: \qquad d = 11-5-0 = 6 \,, \\
&\mathcal D^+(\text{II})|_{V_3>0}: \qquad d = 11-5-1 = 5 \,, \\
&\mathcal D^0(\text{II}): \qquad \qquad \; d = 11-6-1 = 4 \,.
\end{align}
In both branches of $C(\text{II})$ the redundant degree of freedom corresponds to a frame rotation around $\mathbf e_1$. In $\mathcal D^+(\text{II})|_{V_3=0}$ the gauge is fixed completely ($r=0$) since a rotation around $\mathbf e_1$ or $\mathbf e_2$ is incompatible with constraint equations (\ref{C4}) and (\ref{C2}), respectively, whereas a rotation around $\mathbf e_3$ is incompatible with the gauge fixing condition $V_2=0$. 

The dimension of the sets $\mathcal D^+(\text{II})|_{V_3>0}$ and $\mathcal D^0(\text{II})$ require some further explanation. Note that the defining conditions (\ref{def:DII:1}) imply that the tensors in $\mathcal D^+(\text{II})|_{V_3>0}$ have the form
\be
\Sigma_{ab} = \begin{bmatrix}
	* & 0 & * \\
	0 & * & 0 \\
	* & 0 & * 	
\end{bmatrix} , \quad 
N_{ab} = \begin{bmatrix}
	0 & 0 & 0 \\
	0 & * & 0 \\
	0 & 0 & 0 	
\end{bmatrix} , \quad 
V_{a} = \begin{bmatrix}
* \\
0 \\
*  	
\end{bmatrix} ,
\label{matrix-representation}
\ee
where an asterix denotes a non-zero component, and similar in $\mathcal D^0(\text{II})$ with the exception that $V_{a}=\vec{0}$. Therefore, the redundant degree of freedom actually corresponds to a frame rotation around $\mathbf e_2$, and not around $\mathbf e_1$ like in $\mathcal C(\text{II})$. The existence of two distinct gauge transformations in Bianchi type II is a consequence of the fact that $N_{ab}$ has two zero-eigenvalues; there are actually two distinct pairs of Killing vectors corresponding to the $G_2$ subgroup.

\subsubsection{Bianchi type III \label{subsub:III}} 
$\mathcal C(\text{III})$ does not cover all orbits of Bianchi type III. Additionally, the following imperfect branch set is needed:
\begin{align}
\mathcal D^+(\text{III}):& \qquad V_3>0 \,, \quad N_\times \Sigma_3 + \Theta V_3 = 0\,, \quad (\ref{C3}) \,, \quad (\ref{common:DIII1}) \,, \quad (\ref{common:DIII2}) \,,  \label{def:D+III}  
\end{align}
where  
\begin{align}
A &= \sqrt{3} N_\times > 0 \,, \label{common:DIII1} \\ 
N_+ &= \sqrt{3} N_- \,.  \label{common:DIII2} 
\end{align}
Note that given (\ref{common:DIII1}) and (\ref{common:DIII2}) the defining equation (\ref{def:h}) for the parameter $h$, uniquely gives $h=-1$. Thus the set is of Bianchi type VI$_{-1}=\text{III}$. Bianchi type III universes with $V_3>0$ necessarily belongs to $\mathcal D^+(\text{III})$.  Other Bianchi type III universes, with $V_3=0$, belongs to $\mathcal{C}(\text{III}) \equiv \mathcal{C}^+(\text{III}) \cup \mathcal{C}^0(\text{III})$. In particular:\footnote{Note that the limit in (\ref{C+III:subset}) is an invariant subset itself.}
\be
\lim\limits_{V_3\rightarrow 0} \mathcal D^+(\text{III}) \subset  \mathcal C(\text{III})  \,.
\label{C+III:subset}
\ee 
However, the two sets in (\ref{C+III:subset}) are physically equivalent, as proved by Lemma \ref{Lemma:spectrum:III} below. Specifically, $\mathcal C(\text{III})$ occupies a larger region in state space only because of the presence of a gauge mode. Thus (\ref{common:DIII1})-(\ref{common:DIII2}) provide an example on how the gauge can be fixed in the sets $\mathcal C^+(\text{III})$ and $\mathcal C^0(\text{III})$.  
\begin{lemma}
	The sets $\lim\limits_{V_3\rightarrow 0} \mathcal D^+(\text{III})$ and $\mathcal C(\text{III})$ are equivalent upon frame rotations.
	\label{Lemma:spectrum:III}
\end{lemma}
\begin{proof}
	First calculate the limit and compare the defining conditions:
	\begin{align}
	\lim\limits_{V_3\rightarrow 0} \mathcal D^+(\text{III}):& \qquad V_3=0 \,, \quad \Sigma_3 = 0\,, \quad (\ref{C3}) \,, \quad (\ref{common:DIII1}) \,, \quad (\ref{common:DIII2}) \,, \label{endef:1} \\
	\mathcal C(\text{III}):& \qquad V_3=0 \,, \quad \Sigma_3 = 0\,, \quad (\ref{C3}) \,, \quad (\ref{type:III}) \,. \label{endef:2}
	\end{align}
	Then use the transformation properties in Appendix \ref{app:transformations} to show that condition (\ref{type:III}) takes the form (\ref{common:DIII1}) and (\ref{common:DIII2}) in a particular frame. Specifically, the Bianchi type III defining condition (\ref{type:III}) can be written
	\be
	A^2+N_+^2 = 3N_-^2 + 3N_\times^2  \,.
	\label{hei23}
	\ee
	The variables on the left-hand side are spin-0 objects under a frame rotation around $\mathbf e_1$, whereas $[N_-, N_\times]^T$ is a spin-2 object. It follows that 
	\be
	3N_\times^2 \in [0,\; A^2 + N_+^2] \,.
	\ee 
	Thus we can always fix the gauge, up to a reflection $[N_-, N_\times]^T \rightarrow [-N_-, N_\times]^T$, by choosing a frame so that $\sqrt{3} N_\times=A$. This is identical to (\ref{common:DIII1}). Then (\ref{hei23}) gives $N_+=\pm\sqrt{3}N_-$. By employing the remaining reflection, which corresponds to a remaining discrete frame rotation around $\mathbf e_1$, one can always choose $N_+=\sqrt{3}N_-$ which is identical to (\ref{common:DIII2}).
\end{proof}

As for the counting of modes, the defining conditions fix the frame uniquely in $\mathcal D^+(\text{III})$ so that $r=0$, whereas $r=1$ in $\mathcal{C}^+(\text{III})$ and $\mathcal{C}^0(\text{III})$. Upon counting constraint equations we obtain the number of independent modes using (\ref{counting:alg}): 
\begin{align} 
\mathcal{D}^+(\text{III}):& \quad d = 11 - 4 - 0 = 7 \,, \\ 
\mathcal{C}^+(\text{III}):& \quad d = 11 - 4 - 1 = 6 \,, \\
\mathcal{C}^0(\text{III}):& \quad d = 11 - 5 - 1 = 5 \,. 
\end{align}

\subsubsection{Bianchi type VI$_{-1/9}$ \label{subsub:exceptional}} 
Defining conditions: 
\begin{align}
\mathcal D^+(\text{VI}_{-1/9}):& \qquad V_1^2 > 0, \quad V_3=0\,, \quad (\ref{C3}) \,, \quad (\ref{common:BVIex1})\,, \quad (\ref{common:BVIex2}) \,, \\
\mathcal D^0(\text{VI}_{-1/9}):& \qquad V_1^2+V_3^2 = 0\,,  \quad (\ref{C3})\,, \quad (\ref{common:BVIex1})\,, \quad (\ref{common:BVIex2}) \,, \label{def:D0VIexp}
\end{align}
where
\begin{align}
&\Sigma_3 \neq 0 \,, \label{common:BVIex1} \\ 
&N_\times = \sqrt{3} A > 0 \,, \quad N_+=\sqrt{3}N_- \,. \label{common:BVIex2} 
\end{align}
This set belongs to Bianchi type VI$_h$ with $h=-1/9$. The set $\mathcal D^0(\text{VI}_{-1/9})$ is often referred to as the ``exceptional case'', because it is the only model with $\sigma_{12}^2+\sigma_{13}^2>0$, and thus the only model in which the group $G_2$ acts 'orthogonally transitively', among orthogonal Bianchi models containing a perfect fluid \cite{bok:EllisWainwright}. In the presence of the scalar field, as we have seen, this is not unique for Bianchi type $\text{VI}_{-1/9}$. Note that $\mathcal C(\text{VI}_{-1/9})$ and $\mathcal D(\text{VI}_{-1/9})$ are disjoint since only the latter has $\Sigma_3\neq0$. 

Note that the gauge is fixed uniquely in both sets; a frame rotation around $\mathbf e_1$ is incompatible with the gauge fixing condition $\Sigma_2=0$. Thus $r=0$ in both $\mathcal D^+(\text{VI}_{-1/9})$ and $\mathcal D^0(\text{VI}_{-1/9})$.  With 4 and 5 constraints on $\mathbb R^{11}$, respectively, the number of independent modes are: 
\begin{align}
\mathcal D^+(\text{VI}_{-1/9}):& \quad d=11-4-0=7 \,, \\
\mathcal D^0(\text{VI}_{-1/9}):& \quad d = 11 - 5 - 0 = 6 \,. \label{count:D0VIexp}
\end{align}
This is one more than the corresponding sets with $\Sigma_3=0$ (see footnote \ref{fn:ex:c} for details):
\begin{align}
\mathcal C^+(\text{VI}_{-1/9}):& \quad d=11-4-1=6 \,, \\
\mathcal C^0(\text{VI}_{-1/9}):& \quad d = 11 - 5 - 1 = 5 \,.
\end{align}
Along the same line as the proof of Lemma \ref{Lemma:spectrum:III}, the gauge mode of $\mathcal C(\text{VI}_{-1/9})$ can be removed by choosing the gauge conditions (\ref{common:BVIex2}). In this gauge the distinction between the sets is entirely in the variable $\Sigma_3$ which is zero in $\mathcal C(\text{VI}_{-1/9})$ and non-zero in $\mathcal D(\text{VI}_{-1/9})$. A gauge-invariant, and physical, consequence of this is that $\mathcal D^+(\text{VI}_{-1/9})$ contains a spatial field strength one-form $X_a$ that rotates relative to an inertial frame, whereas in $\mathcal C^+(\text{VI}_{-1/9})$ the direction of $X_a$ is fixed relative to the spin axis of a gyroscope.

\section{Higher symmetries \label{ch:higher}}
The goal of this section is to derive all subsets of state space that represent higher-symmetry models. This include locally rotationally symmetric (LRS) models as well as isotropic (FLRW) models, which enjoy isotropy subgroups of dimension 1 and 3, respectively. Higher-symmetry models with imperfect matter are found in the same spacetimes as LRS models with perfect fluids. The latter class is well understood \cite{bok:EllisWainwright} and we start with a brief summary of well-known results:
\begin{enumerate}
	\item FLRW models have representations in Bianchi type I (flat), V (open), VII$_0$ (flat) and VII$_h$ with $h>0$ (open) and IX (closed). Such models have isotropic spatial geometry and expansion. By breaking the isotropy of the expansion, these models become LRS models if the shear tensor is axially symmetric.
	\item Other LRS models, where the spatial geometry is anisotropic but symmetric with respect to a fixed axis, are found in Bianchi type II, III and VIII. In these models both the expansion and the spatial curvature are axially symmetric with a common axis.
\end{enumerate}

Below we derive all representations in state space $D$ of higher-symmetry models, for the imperfect matter sector as well as for the perfect matter sector. These subsets will be denoted by the symbol $\mathcal S$. In particular the union of LRS models of a Bianchi type $\#$ are denoted $\mathcal S(\#)$, which is the union of the imperfect and perfect branches:
\be
\mathcal S(\#) = \mathcal S^+(\#) \cup \mathcal S^0(\#) \,.
\ee

See table \ref{tab:sets:LRS} for a summary of the main properties of the obtained higher-symmetry models. Included are also non-LRS models that possess expansion symmetry, derived in section \ref{ch:ex-sym}.  As explained in section \ref{sub:counting}, the algorithm (\ref{counting:alg}) is applicable for higher-symmetry models and has been used to calculate the dimension $d$ of the space of initial data.

\begin{table}[H]
	\footnotesize
	\newcommand\T{\rule{0pt}{2.6ex}}
	\newcommand\B{\rule[-1.2ex]{0pt}{0pt}}
	\begin{center}
\begin{adjustbox}{angle=0}
		\begin{tabular}{ l l l l l l l }
			\hline\hline 
			Notation \T\B   &  Cosmological model & $d$ &  $r$ & symmetry axis & superset     \\ \hline
			\T\B $\mathcal S^+(\text{I})$ & LRS Bianchi I, spatially flat, imperfect, orthogonal & 2  & 1  & $\mathbf n_\text{LRS}\cdot \mathbf e_2 = 0 $ & $\mathcal D^+(\text{I})$  \\					
			\T\B $\mathcal S^0(\text{I})$ & LRS Bianchi I, spatially flat,  perfect, orthogonal & 2  & 2  &  (see text) & $\mathcal D^0(\text{I})$ \\								
			\T\B $\mathcal S^0(\text{I})|_{\Sigma_{ab}=0}$ & flat FLRW, perfect, orthogonal & 1  & 0  &  (isotropic) & $\mathcal S^0(\text{I})$ \\								
			\T\B $\mathcal S^0(\text{II})$ & LRS Bianchi II, perfect, orthogonal & 3  & 1  & $\mathbf n_\text{LRS} \cdot \mathbf e_1 = 0 $ & $\mathcal C^0(\text{II})$  \\								
			\T\B $\mathcal S^+(\text{III})$ & LRS Bianchi III, imperfect, orthogonal & 3  & 0  & $\mathbf n_\text{LRS} \propto \mathbf e_3 $ & $\mathcal D^+(\text{III})$  
			\\											
			\T\B $\mathcal S^0(\text{III})$ & LRS Bianchi III, perfect, orthogonal & 4  & 1  & $\mathbf n_\text{LRS} \cdot \mathbf e_1 = 0$ & $\mathcal C^0(\text{III})$  \\
			\T\B $\mathcal S_{\text{SF}}^+(\text{III})$  & shear-free, LRS Bianchi III, imperfect, orthogonal & 1  & 0  & $\mathbf n_\text{LRS} \propto \mathbf e_3 $ & $\mathcal S^+(\text{III})$  
			\\								
			\T\B $\mathcal S^+(\text{V})$ & LRS Bianchi V, imperfect & 3  & 0  & $\mathbf n_\text{LRS} \propto \mathbf e_1$ & $\mathcal C^+(\text{V})$ \\								
			\T\B $\mathcal S^0(\text{V})$ & open FLRW, perfect, orthogonal & 2  & 0  & (isotropic) & $\mathcal C^0(\text{V})$ \\								
			\T\B $\mathcal S^+(\text{VI}_0)$ & pseudo-LRS Bianchi VI$_0$, imperfect, orthogonal & 3  & 1  & shear axis $\propto \mathbf e_1$ & $\mathcal C^+(\text{VI}_0)$ \\
			\T\B $\mathcal S^0(\text{VI}_0)$ & pseudo-LRS Bianchi VI$_0$, perfect, orthogonal & 3  & 1  & shear axis $\propto \mathbf e_1$  & $\mathcal C^0(\text{VI}_0)$  \\			
			\T\B $\mathcal S^+(\text{VII}_0)$ & LRS Bianchi VII$_0$, similar to $\mathcal S^+(\text{I})$ & 3  & 0  & $\mathbf n_\text{LRS} \propto \mathbf e_1$  & $\mathcal C^+(\text{VII}_0)$  \\								
			\T\B $\mathcal S^0(\text{VII}_0)$ & LRS Bianchi VII$_0$, similar to $\mathcal S^0(\text{I})$ & 3  & 0  & $\mathbf n_\text{LRS} \propto \mathbf e_1$   & $\mathcal C^0(\text{VII}_0)$ \\								
			\T\B $\mathcal S^+(\text{VII}_h)$ & LRS Bianchi VII$_h$, similar to $\mathcal S^+(\text{V})$ & 3  & 0  & (isotropic)  & $\mathcal C^0(\text{VII}_h)$	\\
			\T\B $\mathcal S^0(\text{VII}_h)$ & open FLRW, similar to $\mathcal S^0(\text{V})$ & 2  & 0  & (isotropic)  & $\mathcal C^0(\text{VII}_h)$  \\											
			\hline\hline
		\end{tabular}
	\end{adjustbox}
		\caption{Invariant sets with higher symmetries and properties of the corresponding cosmological model.  $d$ is the dimension of the space of initial data (gauge invariant), which is also the dimension of the corresponding dynamical system. $r$ is the number of unphysical gauge modes. The direction of symmetry axis is indicated, see main text for explicit expressions. } 		
		\label{tab:sets:LRS}
	\end{center}
\end{table}

\subsection{LRS conditions for the imperfect branch \label{sub:LRS_con}}
In order to generalize LRS models to the imperfect matter sector we start by deriving some general conditions common to all higher-symmetry sets $\mathcal S^+(\#)$. In order for a model with $V_a\neq \vec 0$ to be LRS both the dynamics (expansion) and spatial geometry (metric) must be axially symmetric with a common axis aligned with the field strength one-form:
\be
\mathbf n_\text{LRS} \propto V_1 \mathbf e_1 + V_3 \mathbf e_3 \,.
\ee

For concreteness, consider first the case that $\mathbf n_\text{LRS} \propto \mathbf e_1$. The field strength one-form and shear tensor then take the form $V_a=(V_1,\,0,\,0)$ and $\Sigma^a_{\;\;b}=\Sigma_+\cdot\text{diag}(-2,\,1,\,1)$. We immediately conclude that in an arbitrary frame:
\begin{itemize}
	\item $V_a$ is an eigenvector of the traceless matrix $\Sigma^a_{\;\;b}$,
	\item the sum of the three eigenvalues is zero,
	\item there are two distinct eigenvalues, the one associated with $V_a$ has multiplicity 1, and the other one has multiplicity 2.
\end{itemize}
In an arbitrary frame compatible with our gauge choices we obtain, from the above properties, the following LRS conditions for the imperfect branch:
\begin{align}
&\text{LRS conditions for $V_a\neq\vec 0$:}\nonumber\\ & \Sigma_\times = 0 \,, \label{LRS:1}  \\
& 2\Sigma_- V_1 + \Sigma_3 V_3 = 0 \,, \label{LRS:2} \\
& \Sigma_3 V_1 + (\Sigma_-+\sqrt{3}\Sigma_+) V_3 = 0 \,. \label{LRS:3} 
\end{align}
These conditions ensure that $\Sigma^a_{\;\;b}$ is axisymmetric and aligned with $V_a$. Equations (\ref{LRS:1})-(\ref{LRS:3}) are thus necessary conditions for all LRS models with $V_a\neq \vec 0$. The conditions are nevertheless not sufficient conditions for LRS, as exemplified by the pseudo-LRS model considered in section \ref{pseudo}.

\subsection{LRS Bianchi type I \label{sub:LRS_I}}
$\mathcal S^+(\text{I})$ and $\mathcal S^0(\text{I})$ denote the LRS subsets of the Bianchi type I imperfect and perfect branches respectively:
\be
\mathcal S^+(\text{I}) \subset \mathcal D^+(\text{I}) \,, \qquad 
\mathcal S^0(\text{I}) \subset \mathcal D^0(\text{I}) \,.
\ee
Orbits in the imperfect branch $\mathcal C^+(\text{I})$ belong to the LRS subset $\mathcal S^+(\text{I})$ iff the shear tensor is axially symmetric and aligned with $V_a$. The LRS conditions derived in section \ref{sub:LRS_con} are thus sufficient in Bianchi type I. The defining conditions are summarized as
\be
\mathcal S^+(\text{I}): \quad (\ref{type:I})\,,\; (\ref{D+})\,,\; (\ref{LRS:1})\,,\; (\ref{LRS:2})\,,\; (\ref{LRS:3}) \,.
\ee
In $\mathcal S^0(\text{I})$ the matter sector obeys full 3-dimensional rotational invariance so an orbit in $\mathcal D^0(\text{I})$ belongs to $\mathcal S^0(\text{I})$ iff the shear tensor is axially symmetric. The defining condition is thus
\begin{align}
\mathcal S^0(\text{I}): \qquad &\Sigma_{ab} \text{ has minimum two identical eigenvalues} \\
&\text{(LRS Bianchi type I if two, flat FLRW if three)} \nonumber
\end{align} 
on top of the conditions (\ref{type:I}) and (\ref{D0}) for $\mathcal D^0(\text{I})$. Note that if all three eigenvalues are identical, they are necessarily zero since $\Sigma_{ab}$ is trace-free, and, consequently, the model is flat FLRW with a stiff fluid ($\Theta$) in addition to the $\gamma$-law perfect fluid ($\Omega$).

As for the counting of independent modes, first note that $\mathcal S^+(\text{I})$ and $\mathcal S^0(\text{I})$ are defined by 8 and 7 constraint equations, respectively. There remains $r=1$ gauge degree of freedom in $\mathcal S^+(\text{I})$, which represents a frame rotation with respect to $\mathbf e_2$ (this preserves the conditions $V_2=\Sigma_2=\Sigma_\times=0$). In $\mathcal S^0(\text{I})$ there remains $r=2$ gauge modes, for instance $\Sigma_{ab}$ can be diagonalized by setting $\Sigma_\times=\Sigma_3=0$. The number of independent modes are thus (\ref{counting:alg}):
\begin{align}
\mathcal S^+(\text{I}):& \qquad  d = 11 - 8 - 1 = 2 \,, \label{SI:d2} \\
\mathcal S^0(\text{I}):& \qquad  d = 11 - 7 - 2 = 2 \,, \label{SI:d1} 
\end{align}
which accounts for one field strength mode and one shear mode. See section 8.4.1 of \cite{Ben18} for a phase space analysis of $\mathcal S^+(\text{I})$.

\subsection{LRS Bianchi type II \label{sub:LRS_II}}
Let $\mathcal S(\text{II})$ denote the LRS subset of the Bianchi invariant set  $\mathcal B(\text{II})$ defined by equation (\ref{BIS:II}). We find that $\mathcal S(\text{II})$ lays in the subset $\mathcal C^0(\text{II})$, that is, the LRS subset is restricted to the perfect branch:
\be
\mathcal S(\text{II}) = \mathcal S^0(\text{II}) \subset \mathcal C^0(\text{II}) \,, \quad \mathcal S^+(\text{II}) = \emptyset \,.
\ee
These results agree with the analysis in \cite{Mik18} where it was shown that multiple gauge fields are needed in Bianchi type II in order to align the matter axis with the geometrical axis. 

The LRS set is defined by the following conditions:
\begin{align}
& \mathcal S^0(\text{II}): \nonumber \\
& V_1=V_3=0\,, \label{SII:first}\\
& \Sigma_-^2+\Sigma_\times^2 = 3 \Sigma_+^2 \,, \quad \Sigma_3 = 0\,, \label{SII:2} \\ 
& \Sigma_- N_- + \Sigma_\times N_\times = \Sigma_+ N_+ \,, \label{SII:3} \\
& A=0\,, \quad N_+^2 = 3N_-^2 + 3N_\times^2 > 0 \,. \label{SII:last}
\end{align}
The direction of the LRS axis is:   
\be
\mathbf n_\text{LRS} \propto \sin \theta \mathbf e_2 + \cos \theta \mathbf e_3 \,, 
\ee
where $\theta$ is defined by the equation
\be
\Sigma_- = -\sqrt{3} \Sigma_+ \cos (2\theta) \,. 
\ee

To derive the LRS subset one can start by assuming canonical representations where $\mathbf n_\text{LRS}$ is aligned with one of the basis vectors. The possibility $\mathbf n_\text{LRS}\propto\mathbf e_1$ requires $\Sigma_{ab}=\Sigma_+\cdot\text{diag}(-2, 1, 1)$ and $V_a=(V_1, 0, 0)$, which does not define an invariant set in Bianchi type II according to the evolution equation (\ref{Sm}). Thus the LRS axis must be aligned with the orbits of the $G_2$ subgroup. The only possibility for LRS in the imperfect sector is then $\mathbf n_\text{LRS}\propto\mathbf e_3$, but constraint (\ref{C2}) requires $V_3=0$ in Bianchi type II. We conclude that LRS models do not extend into the imperfect sector in Bianchi type II. In the perfect branch, on the other hand, one easily obtain an invariant LRS model by assuming a canonical representation with $\mathbf n_\text{LRS}\propto\mathbf e_3$. By considering all frame rotations compatible with the gauge conditions, one obtain the general defining conditions (\ref{SII:first})-(\ref{SII:last}). In particular, (\ref{SII:2}) ensure that the shear tensor is axisymmetric (two equal eigenvalues) with axis normal to $\mathbf e_1$, whereas (\ref{SII:3}) ensure that the geometrical axis is aligned with the expansion axis.  

Note that the defining conditions permit frame rotations with respect to $\mathbf e_1$ and $\mathbf n_\text{LRS}$. The dimension of the gauge orbit is $r=1$. Taking into account $c=7$ constraints the dimension of the set is: 
\be
\mathcal S^0(\text{II}): \quad d=11-7-1=3 \,. 
\ee
The three modes are associated with the stiff fluid, shear and spatial curvature.

\subsection{LRS Bianchi type III \label{sub:LRS_III}}
In Bianchi type III we find LRS subsets in the imperfect branch as well as in the perfect branch:
\be
\mathcal S^+(\text{III}) \subset \mathcal D^+(\text{III}) \,, \qquad 
\mathcal S^0(\text{III}) \subset \mathcal C^0(\text{III}) \,.
\ee
The defining conditions for the imperfect branch are: 
\begin{align}
&\mathcal S^+(\text{III}):  \nonumber \\
& V_3>0 \,, \quad \Theta=V_1=0\,, \label{def:S+III:1} \\  
& \Sigma_-+\sqrt 3 \Sigma_+ = 0 \,, \quad \Sigma_\times = \Sigma_3 = 0\,, \label{def:S+III:2} \\
& A = \sqrt 3 N_\times > 0 \,, \label{def:S+III:3} \\
&  N_+ = N_- = 0 \,.  \label{def:S+III:4} 
\end{align}
Note that the LRS axis is along the third basis vector: $\mathbf n_\text{LRS}\propto\mathbf e_3$. In invariant terms, the LRS axis is geometrically restricted to be tangent to the orbits of the $G_2$ subgroup, i.e. to be in the span of $\mathbf e_2$ and $\mathbf e_3$, and further restricted by our gauge choice $V_2=0$ to the $\mathbf e_3$ direction. Note that $N_{ab}$ is necessarily trace-free in the imperfect branch: $N^a_{\;\;a}=2N_+=0$.  This follows from the LRS condition (\ref{LRS:1}) in conjunction with the evolution equation (\ref{Sc}) and constraint (\ref{C2}). The subset $\mathcal S^+(\text{III})$ can be derived by combining the LRS conditions (\ref{LRS:1})-(\ref{LRS:3}) with the defining conditions for $\mathcal D^+(\text{III})$. By assuming $\Theta=0$ the defining conditions above follow directly. The case $\Theta\neq 0$ would imply energy flow along the LRS axis $\mathbf n_\text{LRS}$, which is allowed by the LRS symmetry, but in disagreement with constraint equation (\ref{C5}). Conditions (\ref{def:S+III:1})-(\ref{def:S+III:4}) can also be verified by starting with the LRS Bianchi type III metric in a coordinate approach. 

As for the perfect branch the defining conditions are: 
\begin{align}
&\mathcal S^0(\text{III}):  \nonumber \\
& V_1=V_3=0 \,, \label{def:S+III:first} \\   
& \Sigma_-^2 + \Sigma_\times^2 = 3\Sigma_+^2 \,, \quad \Sigma_3=0 \,, \\
& \Sigma_- N_- + \Sigma_\times N_\times = \Sigma_+ N_+ \,, \\
&  \Sigma_+ A + \Sigma_- N_\times - \Sigma_\times N_- = 0 \,, \quad A>0 \,.  
\end{align}
The direction of the LRS axis is again geometrically restricted to the tangent space of the $G_2$ subgroup:   
\be
\mathbf n_\text{LRS} \propto \sin \theta \mathbf e_2 + \cos \theta \mathbf e_3 \,, 
\label{S0:LRS-axis:1}
\ee
where $\theta$ is defined by the equation
\be
\Sigma_- = -\sqrt{3} \Sigma_+ \cos (2\theta) \,. \label{S0:LRS-axis:2}
\ee

As for the gauge fixing we first note that since $A>0$, the only possible frame transformation is a rotation around $\mathbf e_1$. In $\mathcal S^+(\text{III})$ the gauge is fixed uniquely, $r=0$, since such a frame rotation is incompatible with the gauge fixing condition $V_2=0$. In $\mathcal S^0(\text{III})$ a frame rotation around $\mathbf e_1$ is possible: $r=1$. Counting restrictions we find $c=8$ for $\mathcal S^+(\text{III})$ and $c=6$ for $\mathcal S^0(\text{III})$. The dimensions are thus:
\begin{align}
\mathcal S^+(\text{III}):& \quad  d = 11 - 8 - 0 = 3 \,, \\
\mathcal S^0(\text{III}):& \quad  d = 11 - 6 - 1 = 4 \,,
\end{align}
which accounts for one field strength mode, one shear mode and one spatial curvature mode. The ``extra'', fourth, mode in $\mathcal S^0(\text{III})$ accounts for the trace of $N_{ab}$, which is not necessarily zero in the perfect branch.   To see the correspondence between the two sets, consider the defining conditions of $\mathcal S^0(\text{III})$ in a diagonal shear frame where $\Sigma_\times=0$. In that case $\mathbf n_\text{LRS}$ is aligned with $\mathbf e_2$ or $\mathbf e_3$. In the latter case  the defining conditions for $\mathcal S^0(\text{III})$ match those of $\mathcal S^+(\text{III})$ with $V_3 \rightarrow 0$, but with $\Theta$ and $N_+=-\sqrt{3} N_-$ being free in the perfect branch. Note that $N_+=-\sqrt{3} N_- \neq 0$ is possible only in the perfect branch according to constraint (\ref{C2}).

\subsection{LRS Bianchi type V \label{sub:LRS_V}}

In Bianchi type V we find an LRS subset in the imperfect branch: 
\be
\mathcal S^+(\text{V}) \subset \mathcal C^+(\text{V}) \,, \qquad 
\ee
which is defined by the following conditions: 
\begin{align}
\mathcal S^+(\text{V}): \quad &V_1^2>0\,, \;\; V_3=0\,, \;\; \Sigma_-=\Sigma_\times=\Sigma_3=0 \,,  \;\; A > 0, \;\;  N_+ = N_-= N_\times = 0\,, \\ 
 &A\Sigma_+=\Theta V_1 \,. \label{SV+:2}
\end{align}
This subset is derived from the defining conditions for $\mathcal C^+(\text{V})$ and the LRS conditions (\ref{LRS:1})-(\ref{LRS:3}). The LRS axis is $\mathbf n_\text{LRS} \propto \mathbf e_1$.  There is energy flux along the LRS axis iff $\Theta\neq0$. 

As for the counting of modes, $A>0$ implies that the only possible frame rotation is around the LRS axis. The state vector is invariant under such rotations so the gauge orbits have dimension $r=0$.  Counting constraints we find $c=8$. Thus the dimension of the set is:
\be
\mathcal S^+(\text{V}): \quad d=11-8-0 = 3\,.
\ee  
There are two variables associated with the field strength ($\Theta$, $V_1$), one with shear ($\Sigma_+$) and one with spatial curvature ($A$). The four variables are related by (\ref{SV+:2}), thereby leaving three independent modes. 

As for the perfect matter sector the LRS subset of the perfect branch
\be
\mathcal S^0(\text{V}) \subset \mathcal C^0(\text{V}) \,, \qquad 
\ee
is actually the open FLRW model. The defining conditions are written down in section \ref{sub:FLRW}. To see this note that (\ref{C3}), (\ref{type:V}) and (\ref{def:C0}) imply $\Sigma_+=\Sigma_3=0$. The eigenvalues of $\Sigma_{ab}$ are then:
\be
0, \pm \sqrt{3} \sqrt{\Sigma_-^2+\Sigma_\times^2} \,.
\ee
Two of the eigenvalues can only be equal if $\Sigma_{ab}=0$. Thus LRS symmetry implies shear isotropy in the case of Bianchi type V with a perfect fluid. As for the geometry, Bianchi type V spaces are maximally symmetric hyperbolic 3-spaces, and thus shear-free Bianchi type V cosmological models correspond to the open FLRW models.

\subsection{LRS Bianchi type VII \label{sub:LRS_VII}}
We consider class A models ($A=0$) and class B models ($A>0$), separately.  As for class A, the LRS models of Bianchi type VII$_0$ are spatially flat and similar to LRS Bianchi type I models \cite{bok:MacCallum79}. As for class B, LRS models of Bianchi type VII$_h$ have hyperbolic spatial sections and are thus similar to LRS Bianchi type V models \cite{Hewitt-Wainwright:1993}. Thus, the models derived below are type VII representations of the models considered in sections \ref{sub:LRS_I} and \ref{sub:LRS_V}. 
  
\paragraph{Class A} We first consider class A models in which we find LRS models in the imperfect branch as well as in the perfect branch:
\be
\mathcal S^+(\text{VII}_0) \subset \mathcal C^+(\text{VII}_0) \,, \quad \mathcal S^0(\text{VII}_0) \subset \mathcal C^0(\text{VII}_0) \,.
\ee
They are defined by the following conditions:
\begin{align}
&\mathcal S^+(\text{VII}_0): \quad (V_1)^2>0\,, \quad \Theta=0 \,, \quad (\ref{SVII0:common}) \,, \\ &\mathcal S^0(\text{VII}_0): \;\; \quad V_1=0\,, \quad (\ref{SVII0:common}) \,,    
\end{align}
where 
\be
(N_+)^2 > 0\,, \quad A=N_-=N_\times=V_3=\Sigma_-=\Sigma_\times=\Sigma_3=0 \,.
\label{SVII0:common}
\ee 
The LRS axis is normal to the orbits of the $G_2$ subgroup:
\be
\mathbf n_\text{LRS} \propto \mathbf e_1 \,.
\ee
The defining conditions for the imperfect branch $\mathcal S^+(\text{VII}_0)$ follow from the LRS conditions (\ref{LRS:1})-(\ref{LRS:3}) in conjunction with the evolution equations (\ref{Sm}) and (\ref{Sc}) for  $\Sigma_-$ and $\Sigma_\times$ and the defining conditions for $\mathcal C^+(\text{VI}_0)$. 
The defining conditions for the perfect branch $\mathcal S^0(\text{VII}_0)$ are obtained by letting $\Theta$ be free and setting $V_1\rightarrow 0$ in $\mathcal S^+(\text{VII}_0)$. This is in fact the only representation of $\mathcal S^0(\text{VII}_0)$ since the solution is invariant under a frame rotation around $\mathbf e_1$ (LRS axis), whereas a frame rotation around $\mathbf e_2$ or $\mathbf e_3$ gives  $N_{1a}\neq \vec{0}$, which is incompatible with our gauge choice. 

As for the counting, it follows from the last paragraph above that $r=0$ for both sets. Counting constraints we find $c=8$ for both sets. The dimension of the space of initial data is thus:
\begin{align}
\mathcal S^+(\text{VII}_0):& \quad  d = 11 - 8 - 0 = 3 \,, \\
\mathcal S^0(\text{VII}_0):& \quad  d = 11 - 8 - 0 = 3 \,.
\end{align}
This is one higher than LRS models of Bianchi type I, which accounts for the presence of the variable $N_+$, which in this case is unphysical and whose only role is to define the type VII$_0$ algebra.\footnote{The Bianchi classification, which is based on the Lie algebra of the Killing vectors, is not a solution to the equivalence problem \cite{bok:EllisMaartensMacCallum}.}  

\paragraph{Class B} In class B we consider first the imperfect branch in which we find an LRS subset  
\be
\mathcal S^+(\text{VII}_h) \subset \mathcal C^+(\text{VII}_h) \,, 
\ee
defined by the following conditions:
\begin{align}
\mathcal S^+(\text{VII}_h): \quad &V_1^2>0\,, \quad A^2=hN_+^2>0 \,, \quad \Theta V_1 = \Sigma_+ A \,, \\ 
& V_3 = \Sigma_- = \Sigma_\times = \Sigma_3 = N_-=N_\times = 0 \,.
\end{align}
As in class A the LRS axis is along $\mathbf e_1$. To derive the LRS subset, start with the defining conditions for the superset  $S^+(\text{VII}_h)$. It follows from the LRS conditions (\ref{LRS:1})-(\ref{LRS:3}) that $\Sigma_-=0$ and $\Sigma_\times=0$. In order for these conditions to be preserved in time it follows from the evolution equations (\ref{Sm}) and (\ref{Sc})  for $\Sigma_-$ and $\Sigma_\times$ that 
\be
\begin{bmatrix}
	A & N_+  \\
	-N_+ & A  	
\end{bmatrix}  
\begin{bmatrix}
	N_-  \\
	N_\times  	
\end{bmatrix}  
= \begin{bmatrix}
	0 \\
	0   	
\end{bmatrix} \,,
\ee  
which implies $N_-=N_\times=0$ since the determinant $A^2+N_+^2>0$. The counting of degrees of freedom is similar as for $\mathcal S^+(\text{V})$:
\begin{align}
\mathcal S^+(\text{VII}_h):& \quad  d = 11 - 8 - 0 = 3 \,. 
\end{align}
Note from the evolution equations that $\mathcal S^+(\text{VII}_h)$ and $\mathcal S^+(\text{V})$ evolve equivalently, as they should. Both posses energy-flux along the LRS axis iff $\Theta\neq0$.

As for the perfect branch the LRS subset of $\mathcal C^0(\text{VII}_h)$ is the open FLRW model. The defining conditions are given in subsection \ref{sub:FLRW}, and are derived by taking the limit $V_1\rightarrow 0$ of $\mathcal S^+(\text{VII}_h)$.

\subsection{FLRW models \label{sub:FLRW}}
State space contains flat and open FLRW models, whereas closed FLRW models belong to Bianchi type IX and thus fall outside our considered range of models. 

Spatially flat LRS models have representations in Bianchi type VII$_0$ as well as in type I. Thus there are two representations in state space for flat FLRW models, which are the shear-free subsets of $\mathcal S^0(\text{I})$ and $\mathcal S^0(\text{VII}_0)$:
\begin{align}
\mathcal S^0(\text{I})|_{\Sigma_{ab}=0}:& \quad N_+=0\,, \quad \Omega+\Theta^2=1 \quad (\text{flat FLRW}) \,, \\
\mathcal S^0(\text{VII}_0)|_{\Sigma_{ab}=0}:& \quad N_+^2>0\,, \quad \Omega+\Theta^2=1 \quad (\text{flat FLRW}) \,. 
\end{align}

As for open FLRW there are also two distinct representations in state space. These are the LRS subsets of $\mathcal C^0(\text{V})$ and $\mathcal C^0(\text{VII}_h)$: 
\begin{align}
\mathcal S^0(\text{V}):& \quad A>0 \,, \quad N_+=0\,, \quad\Omega+\Theta^2+A^2=1 \quad (\text{open FLRW})\,, \\
\mathcal S^0(\text{VII}_h):& \quad A^2 = hN_+^2 > 0  \,,\, \qquad \Omega+\Theta^2+A^2=1 \quad (\text{open FLRW}) \,.
\end{align}
The three-dimensional Ricci scalar is ${^3}\!R=-6a^2$ in both representations. 

Due to isotropy frame rotations leave all state vectors invariant: $r=0$. Upon counting constraints we find $c=10$ in flat FLRW, and $c=9$ in both sets of open FLRW. Thus the dimensions are:
\begin{align}
\mathcal S^0(\text{I})|_{\Sigma_{ab}=0}:& \quad  d = 11 - 10 - 0 = 1 \quad (\text{flat FLRW}) \,, \\
\mathcal S^0(\text{VII}_0)|_{\Sigma_{ab}=0}:& \quad  d = 11 - 9 - 0 = 2 \quad (\text{flat FLRW}) \,, \\
\mathcal S^0(\text{V}):& \quad  d = 11 - 9 - 0 = 2 \quad (\text{open FLRW}) \,, \\
\mathcal S^0(\text{VII}_h):& \quad  d = 11 - 9 - 0 = 2 \quad (\text{open FLRW}) \,.
\end{align}
In all cases one mode is associated with the stiff fluid ($\Theta$). In $\mathcal S^0(\text{VII}_0)|_{\Sigma_{ab}=0}$ the variable $N_+$, which defines the Lie algebra, represents an unphysical mode. In open FLRW the second mode is physical and associated with the spatial curvature.

\section{Expansion symmetries \label{ch:ex-sym}}
In models with higher symmetries, as seen above, the metric, shear and energy-momentum tensors share a common rotational symmetry of dimension 1 or 3, which correspond to LRS and FLRW models, respectively. One may now question: 
\begin{enumerate}
	\item Are models with axisymmetric expansion restricted to LRS models? 
	\item Are models with isotropic expansion restricted to FLRW models? 
\end{enumerate}
It is easy to believe, perhaps, that the answers are \emph{yes} and \emph{yes}, since the shear evolution equations are sourced by matter anisotropies and geometrical anisotropies.  

As for the first question, the answer is actually \emph{no} even in orthogonal Bianchi models containing only a perfect fluid, as pointed out already in \cite{main:Ellis1969}. In subsection \ref{pseudo} we show that this type of model, which we refer to as 'pseudo LRS', extends into the imperfect matter sector. 

As for the second question, the answer is \emph{yes} only if attention is restricted to perfect fluid models. In models with imperfect matter, on the other hand, the anisotropic stress may actually cancel the anisotropic curvature allowing for a shear-free normal congruence, as first pointed out in \cite{main:Mimoso93}. Such models have been further developed and studied in \cite{main:Coley94,main:Coley94nr2,main:Carneiro01, main:Koivisto11,pert:Pereira12,shear-free:Abebe16,pert:Pereira17,Mik18,Pailas:2019}. In subsection \ref{sub:shear-free} we show that models, beyond FLRW, with shear-free normal congruence belong uniquely to the set $\mathcal S^+(\text{III})$.

\subsection{Pseudo-LRS Bianchi type VI$_0$ \label{pseudo}}
From the study of Bianchi models with a perfect fluid \cite{main:Ellis1969} it is known that Bianchi type VI$_0$ with trace-free matrix $n_{ab}$ is the unique case of a non-LRS cosmological model that allows expansion symmetry relative to a fixed axis. Since the geometry of spatial sections break the axial symmetry, we refer to such models as `pseudo-LRS'. We shall see that the model extends into the imperfect branch given that the energy-momentum tensor (which is also axisymmetric) is aligned with the expansion symmetry. 

The invariant sets corresponding pseudo-LRS models of Bianchi type VI$_0$ are denoted by $\mathcal S^+(\text{VI}_0)$ for the imperfect branch and by $\mathcal S^0(\text{VI}_0)$ for the perfect branch. They are subsets of the Bianchi type VI$_0$ invariant sets defined in section \ref{ch:invariant}:
\be
\mathcal S^+(\text{VI}_0) \subset \mathcal C^+(\text{VI}_0) \,, \quad \mathcal S^0(\text{VI}_0) \subset \mathcal C^0(\text{VI}_0) \,,
\ee
and defined by the following conditions:
\begin{align}
&\mathcal S^+(\text{VI}_0): \quad (V_1)^2>0 \,, \quad \Theta=V_3=0 \,, \quad (\ref{S0:common}) \,, \\
&\mathcal S^0(\text{VI}_0): \;\; \quad V_1=V_3=0\,, \quad (\ref{S0:common}) \,, 
\end{align}
where the common conditions are
\be
\Sigma_-=\Sigma_\times=\Sigma_3=0\,, \quad N_\times^2 + N_-^2 > 0\,, \quad A=N_+=0 \,.
\label{S0:common}
\ee
In the imperfect branch the axis of the shear-tensor is aligned with the field strength one-form:  
\be
\Sigma_{ab} = \Sigma_+ \cdot \text{diag}(-2,\;1,\;1) \,, \quad V_a = (V_1\,, 0 \,, 0) \,.
\ee 

To derive the defining conditions above it is useful to note that the condition $N_+=0$ (trace-free $n_{ab}$) is preserved in time iff 
\be
\Sigma_- N_- + \Sigma_\times N_\times = 0 \,. 
\label{S0:useful}
\ee
The defining conditions for the imperfect branch $\mathcal S^+(\text{VI}_0)$ follows by considering (\ref{S0:useful}) in conjunction with the defining conditions for $\mathcal C^+(\text{VI}_0)$ and conditions (\ref{LRS:1})-(\ref{LRS:3}). The defining conditions for the perfect branch $\mathcal S^0(\text{VI}_0)$ follows directly from (\ref{S0:useful}) in conjunction with the defining conditions for $\mathcal C^0(\text{VI}_0)$. 

As for the gauge fixing we note that the defining conditions are preserved by frame rotations around $\mathbf e_1$, i.e. the symmetry axis of the shear tensor. This is a non-trivial transformation that rotates the spin-2 object $[N_-, N_\times]^T$ associated with spatial curvature. Thus $r=1$ in both $\mathcal S^+(\text{VI}_0)$ and $\mathcal S^0(\text{VI}_0)$. We count $c=7$ constraints in both sets. The dimension of the sets are thus:
\begin{align}
\mathcal S^+(\text{VI}_0):& \quad  d = 11 - 7 - 1 = 3 \,, \\
\mathcal S^0(\text{VI}_0):& \quad  d = 11 - 7 - 1 = 3 \,,
\end{align}
which accounts for one field strength mode, one shear mode and one spatial curvature mode.

\subsection{Shear-free Bianchi type III \label{sub:shear-free}}
Naively one expects that models with a shear-free normal congruence are restricted to FLRW models. However, this is not necessarily the case in models with imperfect matter in which anisotropic stress ($\pi_{ab}$) may cancel anisotropic curvature ($^3S_{ab}$) on the right-hand side of the shear evolution equation (\ref{eq:shear-ev}), as first pointed out in \cite{main:Mimoso93}. As mentioned in the introduction, a physical realization with a free massless scalar field was found in the locally rotationally symmetric (LRS) Bianchi type III metric by Carneiro and Marug\'an  \cite{main:Carneiro01}. This solution represents a peculiar corner of Bianchi models which had not been systematically explored. It is natural to ask how common such types of solutions are. Motivated by this question, in \cite{Mik18} I considered the space of all spatially homogeneous cosmological models (the Kantowski-Sachs metric in addition to Bianchi types I-IX) with a matter sector containing $n$ independent, non-tilted, $p$-form gauge fields ($p\in\{0,1,2,3\}$) with action on the form (\ref{action2}). In this setup it was proved that 1) the LRS Bianchi type III solution is the only spatially homogeneous metric that admits general relativistic solutions with a shear-free normal congruence given that the energy density of the gauge field is positive and 2) this can be realized only by the cases $p=0$ and $p=2$ which are equivalent and correspond to a free massless scalar field. These results hold for all $n$. Recently, Pailasa and Christodoulakis considered the case of an electromagnetic field ($p=1$) interacting with a charged perfect matter fluid (beyond the assumptions considered in \cite{Mik18}) and showed that a shear-free normal congruence is possible in Bianchi type III, VIII and IX \cite{Pailas:2019}. 

Below, equipped with the dynamical system (\ref{O})-(\ref{C6}), the results of \cite{Mik18} for the case of a single ($n=1$) $p$-form gauge field, with $p\in\{0,2\}$, will be reproduced. This gives an independent verification based on the orthonormal frame approach, whereas the metric approach was employed in \cite{Mik18}. Furthermore, we establish the neighbourhood of the resulting shear-free set within state space $D$. Remind that the analysis here excludes Bianchi models of type VIII and IX, but in these cases the only allowed component of the field strength one-form is the temporal component, as shown in \cite{Ben18}. Specifically, constraint equations give $V_a=0$ for $a=1,2,3$ in Bianchi type VIII and IX, which means that the scalar field corresponds to a non-tilted stiff fluid with vanishing anisotropic stress. This excludes the possibility of shear-free solutions in type VIII and IX, beyond the FLRW subclass of the latter, with our considered matter model.

In order to find all solutions with shear-free normal congruence we set $\Sigma_{ab}=0$ in the dynamical system (\ref{O})-(\ref{C6}). The shear evolution equations (\ref{Sp})-(\ref{S3}) then turn into  constraint equations. Equations (\ref{S3}), (\ref{C3}) and (\ref{C5}) give: 
\begin{equation}
\Theta V_1 = 0\,, \quad \Theta V_3 = 0\,, \quad V_1 V_3 = 0\,.
\end{equation}
Thus two of the variables $\{\Theta, V_1, V_3 \}$ must vanish. If $V_1=V_3=0$ the matter sector is isotropic and there is no shear-free solution beyond FLRW models. If $\Theta=V_3=0$ we have an unused gauge degree of freedom that we fix by choosing $N_\times=0$, i.e. diagonalizing $N_{ab}$. Equation (\ref{Sp}) gives 
\be
N_-^2+V_1^2=0 \, .
\ee 
Again the matter sector is isotropic and, hence, there is no shear-free solution beyond FLRW models. Finally, consider the case $\Theta=V_1=0$. A shear-free solution beyond FLRW requires $\Sigma_{ab}=0$ and $V_3^2>0$. This gives a unique shear-free solution denoted by $\mathcal S_{\text{SF}}^+(\text{III})$ which is a subset of LRS Bianchi III model:
\be
\mathcal S_{\text{SF}}^+(\text{III}) \subset \mathcal S^+(\text{III}) \,.
\ee
The defining conditions follows directly from constraints (\ref{Sm}), (\ref{Sc}), (\ref{C1}), (\ref{C2}), (\ref{C6}): 
\begin{align}
&\mathcal S_{\text{SF}}^+(\text{III}):\nonumber\\
& V_3=\sqrt{2}N_\times>0 \,, \quad \Theta=V_1=0\,, \\
&\Sigma_+=\Sigma_-=\Sigma_\times=\Sigma_3=0\,, \\
&A = \sqrt{3} N_\times>0\,, \quad N_+=0\,, \quad N_-=0 \,.
\end{align}
Note that the gauge is fixed uniquely ($r=0$), as it is for the superset $\mathcal S^+(\text{III})$. With $c=10$ constraints the dimension of the shear-free set is:
\be
\mathcal S_{\text{SF}}^+(\text{III}):\quad  d = 11 - 10 - 0 = 1 \,.
\ee
Note that the Hamiltonian constraint (\ref{C6}) can be written as
\be
\Omega+6N_\times^2= 1 \,,
\ee
whereas the remaining constraint equations are identities.
The deceleration parameter reduces to
\be
q=(\tfrac{3}{2}\gamma-1)\Omega \,,
\ee
so the dynamics is equivalent to a an open FLRW solution with effective spatial curvature $\Omega_\text{eff}=6N_\times^2>0$. 

We have thus proved the following theorem: 
\begin{theorem}
	Shear-free solutions of the system (\ref{O})-(\ref{C6}) beyond FLRW models belong uniquely to the LRS Bianchi type III model $\mathcal S^+_\text{SF}(\text{III})$ and are dynamically equivalent to open FLRW models. 
	\label{thm:exact}
\end{theorem}

In a follow-up work the dynamical stability of points in $\mathcal S^+_\text{SF}(\text{III})$ with respect to spatially homogeneous perturbations will be investigated. Since $\mathcal S^+_\text{SF}(\text{III})$ is a subset of the LRS set $\mathcal S^+(\text{III})$, a relevant question is: can a point $X\in \mathcal S^+_\text{SF}(\text{III})$ be perturbed into a non-LRS state vector? More generally, what type of cosmological models can $X+\delta X$ fall into if $X\in \mathcal S^+_\text{SF}(\text{III})$ and $\delta X$ is a spatially homogeneous perturbation? Here $\delta X$ is an infinitesimal displacement such that $X+\delta X$ satisfies constraint equations (\ref{C1})-(\ref{C6}). The answer to these questions are established by the following result:
\begin{lemma}
Let $X\in\mathcal S^+_\text{SF}(\text{III})$. Then  $X + \delta X \in \mathcal D^+(\text{III})$ for all spatially homogeneous perturbations $\delta X$ described by the dynamical system (\ref{O})-(\ref{C6}).  
	\label{lemma:neighbourhood}
\end{lemma}  
\begin{proof}
	Let $\mathbb R^{11}$ denote the extended state space \emph{not} restricted by constraint equations (\ref{C1})-(\ref{C5}) nor the unit ball (\ref{inequality}). As usual $D$ denotes the physical state space, as defined in (\ref{D:def}). Let $\mathcal U$ be an open ball in $\mathbb R^{11}$ with radius $r$ centred at any point in $\mathcal S^+_\text{SF}(\text{III})$. By choosing $r$ sufficiently small ($r=|V_3|$ suffices) it follows from (\ref{C1})-(\ref{C5}) that $\mathcal U \cap D$ satisfies the defining conditions (\ref{def:D+III}) and thus belongs to $\mathcal D^+(\text{III})$. That is $\mathcal U \cap D = \mathcal U \cap \mathcal D^+(\text{III}) \subset \mathcal D^+(\text{III})$ or $\mathcal U \cap \mathcal (D^+(\text{III}))^c=\emptyset$, where $c$ denotes the complement.
\end{proof}

Informally, perturbations around points in $\mathcal S^+_\text{SF}(\text{III})$ always fall into $\mathcal D^+(\text{III})$. As showed, this restriction follows directly from constraint equations (\ref{C1})-(\ref{C5}). Note that the LRS subset $S^+(\text{III})$ accommodates only a part of the perturbations; the neighbourhood of $\mathcal S^+_\text{SF}(\text{III})$ includes both LRS and non-LRS type perturbations.  

\section{Concluding remarks \label{ch:summary}}

In this paper, we systematically investigated the overall structure of the space of Bianchi type I-VII$_h$ cosmological models containing a non-tilted $\gamma$-law perfect fluid and a free and massless scalar field with a spatially homogeneous gradient $\nabla_\mu \varphi$, that generally breaks the isotropy of spatial sections. Our analysis is complete for the imperfect branch since in Bianchi type VIII and IX the ``Maxwell equations'' enforce a perfect energy-momentum tensor. The scalar field is equivalent to a $2$-form gauge field $\boldsymbol{\mathcal A}$ with action of the type (\ref{action2}), whose energy-momentum tensor is conserved. We obtained all possibilities allowed by the field equations in this setup and classified them with respect to the Bianchi type, higher symmetries and the matter content. The main results are summarized in tables \ref{tab:sets:C}, \ref{tab:sets:D} and \ref{tab:sets:LRS}. Below we give a summary of our work and guide how the results can be helpful in the analysis of the dynamics of such models. 
Since studies of more general cosmological models containing imperfect matter are rare we shall give a detailed account on our approach to gauge fixing and how it relates to other works. 

In the orthonormal frame approach the gauge freedom is the 6 dimensional group of Lorentz transformations, which represents the freedom in choice of orthonormal frame. When employed to Bianchi models, that admit a three dimensional group $G_3$ of isometries, a set of preferred orthonormal frames exist in which the timelike frame vector is orthogonal to the tangent vectors of the group orbit. The remaining gauge freedom is an arbitrary (time-dependent) three-dimensional rotation of the spatial frame vectors. The model builder is explicitly equipped with this freedom via the one-form $\Omega_a$(t), that specifies the choice of local angular velocity of the frame vectors $\mathbf e_a$ relative to a Fermi-propagated spatial frame. By choosing three time-dependent functions $\Omega_a$ accompanied by an appropriate choice of initial orientation of the spatial frame vectors, variables can be removed from the dynamical system.  This procedure, that we employed in sections \ref{sub:geometry}-\ref{sub:gauge}, is what we refer to as 'gauge fixing' in the orthonormal frame approach. 

When restricting to Bianchi models of type I-VII$_h$ the $G_3$ admits a two-dimensional Abelian subgroup $G_2$. A convenient 1+1+2 decomposition of Einstein's field equations can then be carried out \cite{bok:EllisWainwright}. In \cite{Ben18} this decomposition was employed to write down the field equations for the model under consideration. In this paper we started with the same approach by choosing an orbit-aligned frame with spatial basis vectors $\mathbf e_2$ and $\mathbf e_3$ tangent to the orbits of the $G_2$ subgroup. This does not fix the frame uniquely, but allows a time dependent three-dimensional rotation of the spatial frame around $\mathbf e_1$. We then carried out a uniform approach to gauge fixing over the space of models. The main new technical result in section \ref{ch:model} is that the remaining gauge freedom can be employed to remove \emph{two} degrees of freedom from the dynamical system without loss of generality, that is without restricting the space of physical models. Specifically, taking into account constraint equations, the objects $[V_2, V_3]^T$ and $[\Sigma_2, \Sigma_3]^T$, which both transform as spin-1 objects under the remaining frame rotation, can be chosen to be parallel without loss of generality, as established by Lemma \ref{Lemma:gauge}. This allowed us to eliminate $V_2$ and $\Sigma_2$ simultaneously, using the remaining frame rotation around $\mathbf e_1$. The resulting expansion normalized dynamical system is given in section \ref{sub:evo}, and is the starting point for our investigation of cosmological models. Upon eliminating the perfect fluid ($\Omega$) using the Hamiltonian constraint (\ref{C6}), the system consists of 11 variables (subject to 5 non-linear constraints), where $\{\Theta, V_1, V_3\}$ describe the gradient $\nabla_\mu \varphi$, $\{\Sigma_+, \Sigma_-, \Sigma_\times, \Sigma_3 \}$ describe the rate of shear of the perfect fluid and $\{A, N_+, N_-, N_\times \}$ describe the spatial curvature. In section \ref{sub:cauchy} we carried out a powerful check of the equations by explicitly showing that the system posses a well-defined Cauchy problem.

In section \ref{ch:invariant} we wrote state space $D$ as a union of disjoint invariant sets, see equations (\ref{D:subspaces})-(\ref{BIS:VIIh}). Each set was identified with a specific cosmological model with main properties summarized in table \ref{tab:sets:C} and \ref{tab:sets:D}. A perfect branch was obtained, with sets denoted $\mathcal C^0(\#)$ and $\mathcal D^0(\#)$, that reproduced known results for perfect fluid models. Each perfect fluid model $\mathcal C^0(\#)$ have a ``simple'' extension into the imperfect sector denoted $\mathcal C^+(\#)$ in which the direction of $\nabla_i \varphi$ is non-rotating, ie. fixed relative to the spin axis of a gyroscope, and which is orthogonal to the orbits of the $G_2$ subgroup. In Bianchi types I, II, III and VI$_{-1/9}$ there are non-trivial extensions into the imperfect sector, denoted $\mathcal D^+(\#)$, in which $\nabla_i \varphi$ generally rotates in the sense that its direction is not fixed relative to the spin axis of a gyroscope. Unlike the perfect branch, that consists of orthogonal cosmological models, there is generally energy flux in the imperfect models $\mathcal C^+(\#)$ and $\mathcal D^+(\#)$. Note that $A/V_1$ is a constant of motion in all class B models with $V_3=0$. This can be used to simplify the analysis of $\mathcal C^+(\text{III})$, $\mathcal C^+(\text{IV})$, $\mathcal C^+(\text{V})$, $\mathcal C^+(\text{VI}_h)$, $\mathcal C^+(\text{VII}_h)$ and $\mathcal D^+(\text{VI}_{-1/9})$. In section \ref{ch:higher} we derived all subsets corresponding to higher-symmetry models, that is LRS and FLRW models, see table \ref{tab:sets:LRS} for summary of main properties. It was found that LRS models extend into the imperfect branch in all Bianchi types except II. All LRS models are orthogonal cosmological models except for $\mathcal S^+(\text{V})$ and $\mathcal S^+(\text{VII}_h)$ (which are physically equivalent) in which there is energy flux along the LRS axis. In section \ref{ch:ex-sym} we derived all models with a shear-free normal congruence, which was shown to be restricted to the LRS Bianchi type III model, $\mathcal S_\text{SF}^+(\text{III})\subset \mathcal S^+(\text{III})$, verifying previous results \cite{Mik18}. Furthermore, we established $\mathcal D^+(\text{III})$ as a neighbourhood of $\mathcal S_\text{SF}^+(\text{III})$, which means that points $X\in\mathcal S_\text{SF}^+(\text{III})$ can be perturbed to a non-LRS state vector. This will be used in a follow-up to investigate the stability of shear-free solutions with respect to arbitrary spatially homogeneous perturbations.

For each set we counted the degrees of freedom. An algorithm to obtain the dimension $d$ of the space of initial data for each set was presented in section \ref{sub:counting}. This is the number of independent parameters that must be specified to fix initial conditions and is a gauge independent measure of the generality of the corresponding cosmological model. Since we fixed the gauge from the start, no redundant (gauge) degrees of freedom are present ($r=0$) in the most general sets, which are $\mathcal D^+(\text{III})$ and $\mathcal D^+(\text{VI}_{-1/9})$, both of dimension $d=7$. But in the less general sets unphysical gauge degrees of freedom often remain, typically $r=1$. The reason is that the anisotropies used to choose a unique orthonormal frame are often absent in the less general models. This is the situation, for instance, in all sets of type $\mathcal C(\#)$, where $V_i$ and $A_i$ are parallel if both present. The implications for the dynamical system approach depends on the context. In the case that one wants to study the dynamics of a cosmological model that corresponds to a set with $r>0$, the gauge must be fixed further. If not the dynamical system would be of dimension $d+r$ and contain unphysical modes. The remaining gauge fixing is usually a quick job and we have given several examples in the text. In the case that one wants to study a more general cosmological model the situation may be more complicated. Note that the less general sets with $r>0$ often represent the boundary of the more general sets with $r=0$. We emphasize that it is not always possible to fix the gauge uniquely on the boundary without restricting the space of physical models in the interior. Informally, there is not enough physical modes left in the more special sets to cover the entire boundary of the more general set. This is a consequence of the multiple non-linear constraint equations (eqs. \ref{C1}-\ref{C5}) that gives state space $D$ a very complicated topology. 

This has profound consequences for the dynamical system approach since self-similar solutions on the boundary of state space often represent asymptotic states of the more general cosmological model \cite{bok:EllisWainwright}. Because of the redundant gauge degree of freedom, such self-similar solutions will have a non-trivial representation on the boundary of the more general cosmological model. Typically a single physical solution will be represented by a curve of physically equivalent points on the boundary of the set, despite the gauge being fixed uniquely in the interior. This has to be dealt with carefully in the stability analysis of self-similar solutions that lays on the boundary. The reason is that perturbations along the curve into the interior of state space, where the gauge is fixed uniquely, will represent different physical situations. Consequently, the stability will generally change along the curve of physically equivalent points. Giving concrete examples on how to deal with this situation is left for future work. In order to keep the results of this work directly applicable in a wide range of situations, we have formulated the defining conditions of each set general enough to account for all representations allowed by our choice of orthonormal frame.

\paragraph{Acknowledgments}
I would like to thank Sigbj\o rn Hervik, Ben David Normann and Angelo Ricciardone for many interesting conversations on the topic of this paper.

\appendix

\section{Expansion normalized variables \label{app:normalized}}
\begin{align}
N_{ab} &\equiv \frac{n_{ab}}{H} =
\left(\begin{array}{ccc}
0      & 0&0\\
0     & N_++\sqrt{3}N_-&\sqrt{3}N_\times\\
0& \sqrt{3}N_\times &N_+-\sqrt{3}N_-\\
\end{array}\right) , \label{matrix:N} \\
A_i &\equiv \frac{a_i}{H} = (A, 0, 0) \,, \\
\Sigma_{ab} &\equiv \frac{\sigma_{ab}}{H} =
\left(\begin{array}{ccc}
-2\Sigma_+      & \sqrt{3}\Sigma_2&\sqrt{3}\Sigma_3\\
\sqrt{3}\Sigma_2     & \Sigma_++\sqrt{3}\Sigma_-&\sqrt{3}\Sigma_\times\\
\sqrt{3}\Sigma_3     & \sqrt{3}\Sigma_\times&\Sigma_+-\sqrt{3}\Sigma_-\\
\end{array}\right),
\label{shear-tensor} \\
X_\mu &\equiv \frac{x_\mu}{\sqrt{6}H} = (\Theta, V_1, V_2, V_3) \,,
\label{gauge2}\\
\Omega &\equiv \frac{\rho}{3H^2} \,.
\end{align}
Using (\ref{eq:H-ev}) the deceleration parameter can be expressed in expansion normalized variables:
\be
q\equiv -1-\frac{\dot H}{H^2} = \left(\tfrac{3}{2}\gamma -1\right)\Omega + 2\Theta^2 + 2\left( \Sigma_+^2 +\Sigma_-^2 +\Sigma_\times^2 + \Sigma_2^2 + \Sigma_3^2 \right) \,.
\ee

\section{Transformation properties \label{app:transformations}}
Under a rotation of the spatial frame around the frame vector $\mathbf e_1$ the variables obey the following transformation properties in a matrix representation: 
\begin{itemize}
	\item spin-0 (scalars): $\Omega$, $\Theta$, $V_1$, $\Sigma_+$, $A$ and $N_+$\,. \newline 
	\item spin-1: $\begin{bmatrix}
	V_2 \\ V_3
	\end{bmatrix}$ and $\begin{bmatrix}
	\Sigma_2 \\ \Sigma_3
	\end{bmatrix}$.
	\item spin-2: $\begin{bmatrix}
	\Sigma_- \\ \Sigma_\times
	\end{bmatrix}$ and $\begin{bmatrix}
	N_- \\ N_\times
	\end{bmatrix}$.
\end{itemize}
To be concrete, let the frame rotation be right-handed:
\be
(\mathbf e_1\,,\; \mathbf e_2\,,\; \mathbf e_3) \rightarrow (\mathbf e_1\,,\; \cos \phi \; \mathbf e_2 + \sin \phi \; \mathbf e_3 \,,\; -\sin \phi \; \mathbf e_2 + \cos \phi \; \mathbf e_3) \,.  
\ee
The spin-0 objects (scalars) are then invariants, the spin-1 objects transform as
\be
\begin{bmatrix}
V_2 \\
V_3 
\end{bmatrix}
\rightarrow
\begin{bmatrix}
\cos \phi & \sin \phi \\
-\sin \phi & \cos \phi 
\end{bmatrix}
\begin{bmatrix}
V_2 \\
V_3 
\end{bmatrix}, \qquad
\begin{bmatrix}
\Sigma_2 \\
\Sigma_3 
\end{bmatrix}
\rightarrow
\begin{bmatrix}
\cos \phi & \sin \phi \\
-\sin \phi & \cos \phi 
\end{bmatrix}
\begin{bmatrix}
\Sigma_2 \\
\Sigma_3 
\end{bmatrix}, 
\ee
and the spin-2 objects as
\begin{align}
\begin{bmatrix}
\Sigma_- \\
\Sigma_\times 
\end{bmatrix}
\rightarrow
\begin{bmatrix}
\cos 2\phi & \sin 2\phi \\
-\sin 2\phi & \cos 2\phi 
\end{bmatrix}
\begin{bmatrix}
\Sigma_- \\
\Sigma_\times 
\end{bmatrix}, \qquad
\begin{bmatrix}
N_- \\
N_\times 
\end{bmatrix}
\rightarrow
\begin{bmatrix}
\cos 2\phi & \sin 2\phi \\
-\sin 2\phi & \cos 2\phi 
\end{bmatrix}
\begin{bmatrix}
N_- \\
N_\times 
\end{bmatrix}. 
\end{align}
Also note the following constructed scalars:  
\begin{align}
&(V_2^2+V_3^2)\,, \quad\, (\Sigma_2^2+\Sigma_3^2)\,,\;\; \quad (V_2 \Sigma_2 + V_3 \Sigma_3) \,,\;\;\;\;\, \quad (V_2 \Sigma_3 - V_3 \Sigma_2)\,, \\ 
&(\Sigma_-^2+\Sigma_\times^2) \,, \quad (N_-^2+N_\times^2)\,, \quad (\Sigma_- N_- + \Sigma_\times N_\times)\,, \quad (\Sigma_- N_\times - \Sigma_\times N_-)\,.
\end{align}

\section{Discrete transformations \label{app:discrete}} The gauge fixing conditions (\ref{gauge:R})-(\ref{gauge:S2}) leave a handful of discrete transformations that will now be considered. In section \ref{sub:discrete} this is used to identify discrete symmetries of the dynamical system. First we note that, generally, there is a total of eight different frames in which the covectors $A_i$ and $X_i$ are aligned with the basis vectors according to (\ref{gauge:A}) and (\ref{gauge:V}). These frames are related by the following transformations: 
\be
(\mathbf e_1, \mathbf e_2, \mathbf e_3) \quad \rightarrow \quad (\pm \mathbf e_1, \pm \mathbf e_2, \pm \mathbf e_3) \,.
\label{transformations:8}
\ee 
Among the eight bases, four are left-handed and four are right-handed. We shall now see that the choice (\ref{gauge:R}) of dimensionless frame rotation $R_i$ is preserved only by transformations  that preserve the handedness of the frame. Consider first a half rotation around $\mathbf e_1$
\be
(\mathbf e_1, \mathbf e_2, \mathbf e_3) \quad \rightarrow \quad ( \mathbf e_1, - \mathbf e_2, - \mathbf e_3) \,,
\label{rot:ex}
\ee
under which the frame rotation and shear variables transform as\footnote{The shear variables are off-diagonal components of the normalized shear-tensor, see (\ref{shear-tensor}).}  
\begin{align}
(R_1, R_2, R_3) \quad &\rightarrow \quad (R_1, -R_2, -R_3) \,, \\
(\Sigma_\times, \Sigma_3, \Sigma_2) \quad &\rightarrow \quad (\Sigma_\times, -\Sigma_3, -\Sigma_2) \,,
\end{align} 
which is consistent with the gauge choice (\ref{gauge:R}), and so are half-turns around $\mathbf e_2$ and $\mathbf e_3$ The associated discrete symmetries of the dynamical system will be worked out in section \ref{sub:discrete}. 

Next consider the parity flip
\be
(\mathbf e_1, \mathbf e_2, \mathbf e_3) \quad \rightarrow \quad ( -\mathbf e_1, \mathbf e_2, \mathbf e_3) \,,
\ee
under which the rotation and shear transform as  
\begin{align}
(R_1, R_2, R_3) \quad &\rightarrow \quad (-R_1, R_2, R_3) \,, \\
(\Sigma_\times, \Sigma_3, \Sigma_2) \quad &\rightarrow \quad (\Sigma_\times, -\Sigma_3, -\Sigma_2) \,.
\end{align} 
The gauge choice (\ref{gauge:R}) is clearly not preserved by this transformation. Neither is it by flipping the direction of $\mathbf e_2$ nor $\mathbf e_3$. We note that, generally, reversing the handedness of the frame requires a redefinition of the gauge choice (\ref{gauge:R}); $R_i \rightarrow -R_i$ on its left-hand side. However, in this paper the choice of frame rotation (\ref{gauge:R}) is implemented ``irreversibly'' in the equations of the dynamical system. Therefore, in section \ref{sub:discrete}, we will only consider discrete symmetries associated with transformations of the type (\ref{rot:ex}) that preserve the handedness of the frame.

\bibliographystyle{JHEP}
\bibliography{refs}

\end{document}